%% file: 0-main.tex
\documentclass[a4paper,UKenglish,cleveref, autoref, thm-restate]{lipics-v2021}
\usepackage{times,epsfig,epsf,psfrag,array,amsmath,amssymb,verbatim}
\usepackage{url,comment,enumerate,graphicx,graphics,wrapfig,csquotes}
\usepackage[usenames,dvipsnames,svgnames,table]{xcolor}
\usepackage{multirow}
\usepackage{algorithm}
\usepackage[no end]{algpseudocode}
\usepackage{pdfpages,framed,tcolorbox}
\usepackage{tabularx}
\usepackage{pdflscape,latexsym,bm,array,caption,textcomp}
\usepackage{xspace,xcolor}

\newcommand{\SMPM}{\mbox{{marriage instance with critical men}}}

\newcommand{\PFM}{popular feasible matching}
\newcommand{\mPFM}{minimum size popular feasible matching}
\newcommand{\DFM}{dominant feasible matching}
\newcommand{\NP}{\mbox{{\sf NP}}}
\newcommand{\HRLQ}{\mbox{{\sf HRLQ}}}
\newcommand{\SMone}{\mbox{{$G'$}}}
\newcommand{\SMtwo}{\mbox{{$G''$}}}
\newcommand{\SIAP}{\mbox{{\sf SIAP}}}
\newcommand{\SRAP}{\mbox{{\sf SRAP}}}
\newcommand{\listw}{\mbox{{Pref($w$)}}}
\newcommand{\listm}{\mbox{{Pref($m$)}}}

\title{Popular Edges with Critical Nodes}
\titlerunning{Popular Edges with Critical Nodes}

\author{Kushagra Chatterjee\textsuperscript{}\footnote{\textsuperscript{}{Work was done when the author was a Master's student in Chennai Mathematical Institute}}\footnote[2]{supported in part by an MoE AcRF Tier 2 grant (WBS No. A-8000416-00-00) and an ODPRT grant (WBS No. A-0008078-00-00)}}{National University of Singapore, Singapore }{e0823067@u.nus.edu}{}{}

\author{Prajakta Nimbhorkar\textsuperscript{}\footnote[3]{\textsuperscript{}{}{}{Supported in part by SERB Award CRG/2019/004757. }}}{Chennai Mathematical Institute, India}{prajakta@cmi.ac.in}{}{}

\authorrunning{K. Chatterjee and P. Nimbhorkar}

\Copyright{Kushagra Chatterjee and Prajakta Nimbhorkar}

\nolinenumbers 

\begin{document}

\maketitle

\begin{abstract}
In the {\em popular edge problem}, the input is a bipartite graph $G = (A \cup B,E)$ where $A$ and $B$ denote a set of men and a set of women respectively, and each vertex in $A\cup B$ has a strict preference ordering over its neighbours. A matching $M$ in $G$ is said to be {\em popular} if there is no other matching $M'$ such that the number of vertices that prefer $M'$ to $M$ is more than the number of vertices that prefer $M$ to $M'$. The goal is to determine, whether a given
edge $e$ belongs to some popular matching in $G$. A polynomial-time algorithm for this problem appears in \cite{CK18}.

We consider the popular edge problem when some men or women are prioritized or critical. A matching that matches all the critical nodes is termed as a feasible matching. It follows from \cite{Kavitha14,Kavitha21,NNRS21,NN17} that, when $G$ admits a feasible matching, there always exists a matching that is popular among all feasible matchings. 

We give a polynomial-time algorithm for the popular edge problem in the presence of critical men or women. 
We also show that an analogous result does not hold in the many-to-one setting, which is known as the Hospital-Residents Problem in literature,
even when there are no critical nodes.
\end{abstract}

\input{1-Introduction}

\input{2-mPFM-DFM}
\input{3-Algorithm}

\bibliography{references}
\appendix
\input{4-Appendix}

\end{document}

%% file: 1-Introduction.tex
\section{Introduction}
The stable marriage problem is well-studied in literature. The input instance is a bipartite graph $G=(A\cup B,E)$ where $A$ and $B$ denote the
sets of men and women respectively, and each vertex has a strict preference ordering on its neighbors. The preference ordering is referred to as the
{\em preference list} of the vertex. Such an instance is referred to as a {\em marriage instance}. A matching $M$ in $G$ is said to be {\em stable} if there is no pair $(a,b)\in E$ such that both $a$ and $b$ prefer each 
other over their respective partners in $M$, denoted as $M(a)$ and $M(b)$ respectively. A matching is called {\em unstable} if such a pair $(a,b)$ exists, and such a pair $(a,b)$ is called a {\em blocking pair}. In their seminal paper, Gale and Shapley showed that stable matchings always exist and can be computed in linear time \cite{GS62}. However, all the stable matchings match the same set of vertices \cite{GS85} and they can be as small as half the size of a maximum matching \cite{HK13}. Hence popularity has been considered as an alternative to stability.
\begin{definition}[Popular Matching]
A matching $M$ is said to be {\em popular} in a marriage instance $G$ if, for all matchings $N$ in $G$, the number of 
vertices that prefer $N$ over $M$ is no more than the number of vertices that prefer $M$ over $N$.
\end{definition}
 In other words, $M$ is popular if it does not lose a head-to-head election with any other matching $N$ where votes are cast
by vertices. This notion was introduced by G\"{a}rdenfors \cite{G75} and has been well-studied since then (see Section~\ref{sec:relwork}). Popular matchings always exist since stable matchings are popular, in fact, stable matchings are minimum size popular matchings \cite{HK13}. A subclass of maximum size popular matchings called {\em dominant matchings} was identified in \cite{CK18}. 
\begin{definition}[Dominant Matching]
A matching $M$ in a marriage instance $G$ is called a \emph{dominant matching} if $M$ is popular, and for each $N$ such that $|N| > |M|$, the number of vertices that prefer $M$ to $N$ is more than the number of vertices that prefer $N$ to $M$.
\end{definition}
Informally, a matching $M$ is a dominant matching if $M$ is popular and $M$ wins against any other matching $N$ which is larger than $M$. Note that a dominant matching is clearly a maximum size popular matching but a maximum size popular matching need not be a dominant matching. 

Cseh and Kavitha \cite{CK18} addressed the problem of determining whether there is a popular matching containing a given edge $e$, referred to as the {\em popular edge problem}. They gave a polynomial-time algorithm for this problem. This is surprising since, in \cite{FKPZ19}, it was shown that stable matchings and dominant matchings are the only two tractable subclasses of popular matchings, and it is \NP-hard to find a popular matching which is neither stable nor dominant. 

Popular matchings find applications in situations where certain nodes are prioritized or critical and they are required to be matched. A real-life example of this scenario is assignment of sailors to billets in the US Navy \cite{Robards01,Soldner14,Kavitha21} where certain billets are required to be matched. Rural hospitals often face the problem
of understaffing in the National Resident Matching Program in the USA \cite{Roth84,Roth86}. Thus marking some positions in these hospitals as critical and finding a critical matching provides a way to address this issue.
While matching students to mentors, it may be required to assign mentors to all the students whose past performance is below a certain threshold. In several other applications, a subset of people needs to be prioritized based on their economic, ethnic, geographic, or medical backgrounds.
A matching that matches all the prioritized or {\em critical} nodes is termed as a {\em feasible} matching.
Such a scenario has been considered in \cite{NN17} and \cite{NNRS21} in the many-to-one setting, and it is shown that there always exists a matching that is popular within the set of feasible matchings. 
%Motivated by this, Kavitha \cite{Kavitha21} recently considered the model where a marriage instance $G=(A\cup B,E)$ is given along with a set $C\subseteq A\cup B$ of prioritized nodes, referred to as {\em critical nodes} 
In \cite{Kavitha21}, a matching that matches as many critical nodes as possible has been referred to as a {\em critical matching}. It is shown in \cite{Kavitha21} that a matching that is popular in the set of critical matchings, called a {\em popular critical matching}, always exists and a polynomial time algorithm is given for the same. 
A special case of this is addressed in \cite{Kavitha14}, where all the nodes are critical, and hence a critical matching is a matching that is popular amongst all maximum size matchings. A polynomial-time algorithm is given in \cite{Kavitha14} for this problem. 

In the presence of critical men or women, popular edge problem for feasible matchings is a 
natural question that arises in this context.
Thus, given a marriage instance $G=(A\cup B,E)$, a set of critical nodes $C\subseteq A$, and an edge $e$, the problem is to determine whether there is a feasible matching containing $e$ that is popular within the set of feasible matchings. We call this the {\em popular feasible edge problem}.
\begin{definition}[Popular feasible matching]
Given a marriage instance $G=(A\cup B,E)$, and a set of critical nodes $C$, a feasible matching that is popular among all the feasible matchings is called a {\em popular feasible matching}.
\end{definition}
We also define dominant feasible matchings below.
\begin{definition}[Dominant feasible matching]
Given a marriage instance $G=(A\cup B,E)$ and a set of critical nodes $C$, a matching $M$ is called a {\em dominant feasible matching} if $M$ is a popular feasible matching, and for all the
feasible matchings $N$ such that $|N|>|M|$, $M$ gets strictly more votes than $N$.
\end{definition}

\subsection{Our contributions}

We show the following main result in this paper:
\begin{theorem}\label{thm:main}
Given a marriage instance $G=(A\cup B,E)$ along with a set of critical nodes $C\subseteq A$, an edge $e \in E$ belongs to a popular feasible matching in $G$ if and only if $e$ belongs to a minimum size popular feasible matching or a dominant feasible matching in $G$.
\end{theorem}
Theorem~\ref{thm:main}, along with the following results, leads to a polynomial-time algorithm for the popular critical edge problem.
\begin{theorem}\label{thm:surjective}
There are polynomial-time reductions from a given instance $G$ with a set of critical men $C$ to marriage instances $G'$ and $G''$ such that there is a surjective map
from stable matchings in $G'$ to \mPFM s in $G$ and there is a surjective map from stable matchings in $G''$ to \DFM s in $G$.
\end{theorem}
The reductions are similar to those in \cite{NN17,NNRS21,Kavitha21}, however, the surjectivity of the maps is not shown there. In \cite{Kavitha21icalp}, a similar reduction is given and the surjectivity of the map is shown using dual certificates, whereas our proofs of surjectivity are combinatorial.

{\bf Counter-example for the many-to-one setting: }
We show that a result analogous to Theorem~\ref{thm:main} does not generalize to the many-to-one setting referred to as the {\em Hospital-Residents problem} in literature, even when there are no critical nodes.
Figure~\ref{fig:HR} shows such an example. Informally, popularity in the many-to-one setting is defined as follows. To compare two matchings $M$ and $N$, a hospital casts as many votes as its upper quota. It compares the sets of residents $M(h)$ and $N(h)$ that it gets in the matchings $M$ and $N$ respectively by fixing any correspondence function between $M(h)\setminus N(h)$ and $N(h)\setminus M(h)$. For the formal definition of popularity in the many-to-one setting in the presence of critical nodes, we refer the reader to \cite{NNRS21,NN17}, where it is shown that the respective algorithms output a matching that is popular under any choice of the correspondence function.

\begin{figure}[!h]
    \centering
  \scalebox{0.8}{  \includegraphics[width=10cm, height=6cm]{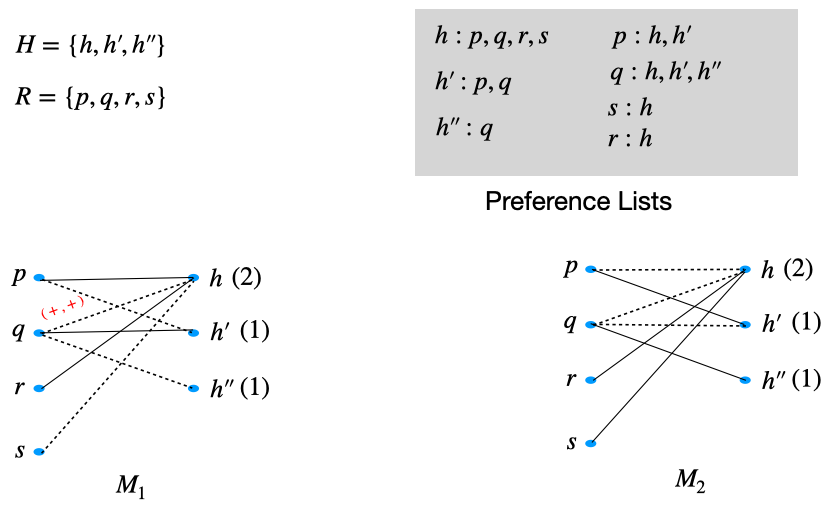}}
    \caption{Here $H$ and $R$ are the sets of hospitals and residents respectively, $h$ has upper quota or capacity $2$, other hospitals have upper quota $1$. The only stable matching is $M=\{(p,h),(q,h)\}$ of size $2$ whereas the only dominant matching is $M_2$, of size $4$. The edge $(q,h')$ belongs to a popular matching $M_1$ of size $3$, but does not belong to the stable matching $M$ or to the dominant matching $M_2$. Thus Theorem~\ref{thm:main}
does not hold for this instance.}
    \label{fig:HR}
\end{figure}

\subsection{Overview of our algorithm}
We give a brief outline of our algorithm.
After proving Theorem \ref{thm:main}, the algorithm to determine whether an edge $e$ belongs to a \PFM\ goes as follows:
\begin{enumerate}[(i)]
\item\label{itm:mpfm} Check whether $e$ belongs to a \mPFM. If so, output yes and stop, otherwise go to the next step.
%then $e$ belongs to a \PFM\ because all mPFMs are PFMs. If NO then go to step (ii).
\item\label{itm:dfm} Check whether $e$ belongs to a \DFM. If so, then output yes and stop. %$e$ belongs to a PFM because all DFMs are PFMs. 
If not, then conclude that $e$ does not belong to a \PFM\ in that instance by Theorem~\ref{thm:main}.%because in this 
\end{enumerate}

For steps (\ref{itm:mpfm}) and (\ref{itm:dfm}) above, we use the reductions mentioned in Theorem~\ref{thm:surjective}. For an edge $e$ in $G$, there are multiple edges in $G'$ and $G''$ corresponding to $e$. The stable edge algorithm of \cite{Knuth76} can be used to determine whether any of the edges that correspond to $e$ in $G'$ or $G''$
is contained in some stable matching in $G'$ or $G''$. The details are given in Section~\ref{sec:main-thm}.

To prove Theorem \ref{thm:main}, we assume that $e$ is contained in a \PFM\ $M$ which is neither a \mPFM\ nor a \DFM. 
We give a \emph{Partition Method} in Section \ref{sec:part} which partitions the given instance
into three parts. We call the restrictions of $M$ on the three parts as $M_d$, $M_m$ and $M_r$. 
Since $e$ is contained in $M$, 
$e$ must belong to one of the three parts viz. $M_d$, $M_m$ and $M_r$. 
If $e \in (M_d \cup M_r)$, we convert the matching $M_m$ to another matching $M'_d$ which is a \DFM\ in that part, and show that 
$M^*_d = (M_d \cup M'_d \cup M_r)$ is a \DFM\ in the whole instance. Thus $e$ is contained in a \DFM, namely $M^*_d$.
Similarly, if $e \in (M_m \cup M_r)$ then we convert the matching $M_d$ to another matching $M_m'$ which is a \mPFM\ in the 
respective part, and moreover, 
$M^*_m = (M_m' \cup M_m \cup M_r)$ is a \mPFM\ in the whole instance.
Thus $e$ belongs to the \mPFM\ $M^*_m$. 

\subsection{Related Results}\label{sec:relwork}
Gale and Shapley proposed  an algorithm to find a stable matching in a marriage instance in their seminal paper \cite{GS62}. The notion of popular matching was introduced by G\"{a}rdenfors \cite{G75}. Popular matchings in the marriage instance have been considered first in \cite{HK13,Kavitha14}. An $O(m)$-time algorithm to find a dominant matching in a marriage instance is given in \cite{Kavitha14}. In \cite{Kavitha14}, a size-popularity tradeoff has been considered, and a polynomial-time algorithm for finding a maximum matching that is popular among all maximum matchings is given.
The popular edge problem is inspired by the \emph{stable edge problem}. The stable edge problem involves deciding whether a given edge $e$ belongs to a stable matching in a Stable Marriage instance. A polynomial-time algorithm for the stable edge problem is given in the book by Knuth\cite{Knuth76}. 

Cseh and Kavitha \cite{CK18} addressed the popular edge problem and gave an $O(m)$ time algorithm for the same.
Later, Faenza et. al \cite{FKPZ19} show that the problem of deciding whether an instance admits a popular matching containing a set of two or more edges is \NP-Hard. In that paper, the authors also show that finding a popular matching in a stable marriage instance, which is neither stable nor dominant is \NP-Hard.

In \cite{NN17}, the authors showed that a \PFM\ always exists in an \HRLQ\ instance. This has been further generalized by Nasre et al. \cite{NNRS21} to the \HRLQ\ case with critical residents. While our work and \cite{NN17,NNRS21} deal with instances that admit a feasible matching, the work of Kavitha \cite{Kavitha21} contains an algorithm to find a popular critical matching i.e., a matching that matches maximum possible number of critical nodes and is popular among all such matchings.  Problems related to \HRLQ\ have also been considered in \cite{HIM16} and \cite{BFIM10} in different settings. 
Besides this, there has been a lot of recent work on various aspects of popular matchings and their generalizations e.g. weighted popular matchings, quasi-popular matchings, extended formulations, popular matchings with one-sided bias, dual certificates to popularity, popular matchings polytope and its extension complexity, hardness and algorithms for popular matchings in case of ties in preferences etc. \cite{Kavitha21icalp,Kavitha20,Kavitha20icalp,KKMSS20,FK20,Kavitha18,HK21,Kavitha16,CHK17,GNNR19}.

{\bf A comparison with \cite{CK18}:} Cseh and Kavitha in their paper \cite{CK18} presented an $O(m)$-time algorithm for the popular edge problem. Our result follows a similar template as theirs, although unlike that in \cite{CK18} where nodes are divided into two levels, we have nodes divided into a number of levels proportional to the number of critical nodes. Also, we need to partition the given instance into three parts, all of which can have blocking pairs, whereas in \cite{CK18}, all the blocking pairs can be put into only one of the two parts straight away. 

\subsection{Organization of the paper}
In Section \ref{sec:algo-mPFM-DFM}, we give the reductions from a \SMPM\ to marriage instances without critical nodes. In Section \ref{sec:main-thm}, we prove Theorem~\ref{thm:main} and discuss the popular edge algorithm.

%% file: 2-mPFM-DFM.tex
\section{The Reductions}\label{sec:algo-mPFM-DFM}
We describe the reductions from a marriage instance $G=(A\cup B,E)$ with a critical node set $C\subseteq A$ to marriage instances $G'$ and $G''$ such that there is a surjective map from
the set of stable matchings in $G'$ to the set of \mPFM s in $G$, and a surjective map from the set of stable matchings in $G''$ to the set of \DFM s in $G$, thereby proving Theorem~\ref{thm:surjective}. 

We recall some notation below, that is standard in popular matchings literature (e.g. \cite{HK13,Kavitha14,CK18} etc.) 
\begin{definition}[Edge labels]
Given a matching $M$ in $G$, a vertex $u$ assigns a label $+1$ (respectively $-1$) to an edge $(u,v)$ incident on it if $(u,v)\notin M$ and $u$
prefers $v$ over its partner in $M$ denoted by $M(u)$ (respectively $M(u)$ over $v$). Thus each edge $(u,v)$ gets two labels, one from $u$ and the other from $v$.  
\end{definition}
By above definition, an edge not present in a given matching $M$ can get one of the four labels $(+1,+1),(+1,-1),(-1,+1),(-1,-1)$. We use the convention that the first label in the pair is from a vertex in $A$ and the second label is from a vertex in $B$. Any vertex prefers to be matched to one of its neighbors over remaining unmatched. 

%-----------------------------------------------------------------------------

\subsection{Reduction from $G$ to $G'$}\label{sec:mPFM-algo}
Given an instance $G$, the instance $G'$ is constructed as follows. Let $C\subseteq A$ be the set of critical nodes and
$\ell=|C|$.

\begin{itemize}
\item \textbf{The set $A'$: }
For each $m \in C$, $A'$ has $(\ell+1)$ copies of $m$, denoted by the set $A'_m=\{m^0,m^1,...,m^{\ell}\}$. We refer to $m^i \in A'$ as the {\em level $i$} copy of $m \in A$. For each $m \in A \setminus C$, $A'$ has only one copy of $m$, denoted $A'_m=\{m^0\}$. Now, $A' = \bigcup\limits_{m \in A} A'_m $.

\item \textbf{The set $B'$: }
All the women in $B$ are present in $B'$. Additionally,
corresponding to each $m \in C$, $B'$ contains $\ell$ dummy women denoted by the set $D_m=\{d_m^1,d_m^2,d_m^3,...,d_m^{\ell}\}$. We call $d_m^i$ as the {\em level $i$ dummy woman for $m$}. For $m \in A\setminus C$, $D_m=\emptyset$. Now, $B' = B \cup \bigcup\limits_{m \in A} D_m$. 
\end{itemize}
We denote by \listm\ and \listw\ the preference lists of $m \in A$ and $w \in B$ respectively. Let \listw$^i$ be the list of level $i$ copies of men present in \listw, if these copies exist. 

We now describe the preference lists in $G'$. Here $\circ$ denotes the concatenation of two lists. 

\begin{center}
		\begin{tabular}{lll}
		$m^0$ s.t. $m\in A\setminus C$  &: & \listm\\		
		$m\in C$, $i \in \{0,\dots ,\ell\}$: & & \\
		$m^0$ &:& \listm, $d_m^1$\\
		$m^i$ &: & $d_m^i$, \listm, $d_m^{i+1}$, $i \in \{1,\ell-1\}$)\\
		$m^{\ell}$ &:& $d_m^{\ell}$, \listm\\
		$w$ s.t. $w\in B$ &: & \listw$^{\ell} \circ$ \listw$^{\ell-1} \circ\ldots\circ $\listw$^0$ \\
		$d^i_m$,$i \in \{1, \dots ,\ell\}$  &:& $m^{i-1}$, $m^i$ \\

		\end{tabular}
	\end{center}
%-----------------------------------------------------------------------------

\subsection{Correctness of the reduction}
After constructing $G'$, the mapping of a stable matching $M'$ in $G'$ to a \mPFM\ $M$ in $G$ is a simple and natural one: For $m \in A$, define the set $M(m) = B \cap \bigcup\limits_{m \in A} M'(m^i)$, which is the set of non-dummy women matched to any copy of $m$ in $A'$. In the rest of this section, the term {\em image} always refers to the image under this map. 

It remains to prove that $M$ is a \mPFM\ i.e., $M$ is a matching in $G$, it is feasible, popular, and no matching smaller than $M$ is popular. 
We define some terminology first. A man $m \in A$ and his matched partner $w\in B$ in $M$ are said to be {\em at level $i$} if $(m^i,w)\in M'$. A man $m \in A$ which is unmatched in $M$ is said to be at level $i$ if $m^i$ in $A'$ is unmatched in $M'$. All unmatched women are said to be at level $0$. Now we give a sufficient condition for a \mPFM\ in $G$.

%------------------------------Theorem 10-------------------------------------
\begin{theorem}\label{thm:mPFM}
The image $M$ of a stable matching $M'$ in $G'$ is a \mPFM\ in $G$ and it satisfies the following conditions. Moreover, any matching $M$ that satisfies the following conditions for some assignment of levels to vertices of $G$ is a \mPFM.
\begin{enumerate}
    \item All $(+1,+1)$ edges are present in between a man at level $i$ and a woman $w$ at level $j$ where $j > i$.
    \item All edges between a man at level $i$ and a woman at level $(i-1)$ are $(-1,-1)$ edges.
    \item No edge is present between a man at level $i$ and a woman at level $j$ where $j \leq (i-2)$, and all the edges of $M$ are between vertices at the same level.
    \item All unmatched men are at level $0$.
\end{enumerate}
\end{theorem}
\begin{proof}
    Here we need to show that if $M$ satisfies the four conditions, then $M$ is a \mPFM. So, to prove this we show at first $M$ is a \PFM\ and then we show that, for all feasible matchings $N$ such that $|N| < |M|$, we have $\phi(N,M) < \phi(M,N)$.

First we prove that $M$ is a feasible matching.
Suppose $M$ is not a feasible matching and there exists a feasible matching $N$ in that \SMPM. Recall that we are only concerned with those instances which have at least one feasible matching. Suppose $m$ is a critical man who is unmatched in $M$. So, the graph $M \oplus N$ must contain an alternating path $\rho$ which starts from $m$. Now $\rho$ can end in a man $m'$ or in a woman $w'$. 

\underline{CASE 1}: $\rho$ ends in $m'$: Let $\rho = (m,w,m_1,w_1,....,m')$. Since $\rho$ ends in $m'$, $m'$ must be unmatched in $N$. Since $N$ is a feasible matching $m'$ must be a non critical man and hence will be at level $0$. Since $m$ is unmatched in $M$, it has to be in the level $\ell$ otherwise if $m$ is at level $i$ where $i < \ell$ then $(m^i,d_m^{i+1})$ would be a $(+1,+1)$ edge in $M'$ because $m^i$ is unmatched in $M'$ and $d_m^{i+1}$ prefers $m^i$ the most in $G'$. Again no woman $w$ which is adjacent to $m$ can be at level $\ell-1$ because then $(m^{\ell},w)$ would form a $(+1,+1)$ edge in $M'$ as $m^{\ell}$ is unmatched and $w$ prefers $m^{\ell}$ more than her matched partner which is at level $\ell-1$. Hence, in $\rho$, $w$ is at level $\ell$ again $M(w) = m_1$ is also at level $\ell$ because the level of a woman and her matched partner are same. Now, $w_1$ cannot be at level less than $(\ell-1)$ due to Condition $3$ of Theorem \ref{thm:mPFM}. Hence the alternating path $\rho$ can go only one level down that is from a man at level $i$ to a woman at level $i-1$. Note that all the men who are at level greater than $0$ are critical men because there is no copy of a non-critical man of level greater than $0$ in $G'$. Since $\rho$ can go only one level down, hence there must exist at least one critical man at each level from $1$ to $\ell-1$ and there are at least two critical men ($m$ and $m_1$) at level $\ell$. Hence, the number of critical men in $G$ is at least $\ell+1$. This is a contradiction because we know the number of critical men in $G$ is $\ell$.

 \underline{CASE $2$}: $\rho$ ends in $w'$. Since $\rho$ ends in $w'$, $w'$ has to unmatched in $M$ and thus the level of $w'$ is $0$ as the level of each unmatched woman is defined to be $0$. Hence $\rho$ starts from a man at level $\ell$ and ends at a woman at level $0$. Since $\rho$ can only go one level down, hence using the same arguments as used in case $1$, we get that there are at least $\ell+1$ critical men in $G$. This is a contradiction because we know the number of critical men in $G$ is $\ell$. Hence $M$ is a feasible matching.

Now, we prove that $M$ is a \PFM.

Consider any feasible matching $N$ in $G$. We need to show that $\phi(N,M) \leq \phi(M,N)$. Consider the graph $M \oplus N$. The graph $M \oplus N$ is a disjoint union of alternating paths and cycles. If $\phi(N,M) > \phi(M,N)$ then at least one of the following three conditions must be satisfied in the graph $M \oplus N$ when we label the edges of $N$ with respect to $M$.
\begin{enumerate}[(a)]
    \item There is an alternating cycle with more $(+1,+1)$ edges than $(-1,-1)$ edges.
    \item There is an alternating path which has at least one end point unmatched in $M$ where the number of $(+1,+1)$ edges is more than the number of $(-1,-1)$ edges.
    \item There is an alternating path which has both the end points matched in $M$ and the number of $(+1,+1)$ edges is at least two more than the number of $(-1,-1)$ edges in that alternating path, and the path ends in a man $m\notin P$.
\end{enumerate}
Now, we show that none of the above conditions are satisfied which implies that $\phi(N,M) \leq \phi(M,N)$ for all feasible matchings $N$ in the \SMPM. Hence $M$ is a \PFM. 

\noindent
\underline{Condition (a)}: From Condition $1$ of theorem \ref{thm:mPFM} we get that a $(+1,+1)$ is present in between a lower level man $m$ and a higher level woman $w$. Let us assume the level of $m$ is $i$ and the level of $w$ is $j$, hence $j > i$. So, if an alternating cycle $\rho$ in $M \oplus N$ has a $(+1,+1)$ edge in between $m$ to $w$ then $\rho$ must return to $m$ again. Now, from the Condition $3$ of theorem \ref{thm:mPFM} we get that an edge in $\rho$ can go only one level down that is from a man at level $i$ to a woman at level $(i-1)$ (not below $(i-1)$) and from condition $2$ we get that all edges in between a man at level $i$ and a woman at level $(i-1)$ are $(-1,-1)$ edges. Hence we get that the alternating subpath of $\rho$ from $w$ to $m$ must contain $(j - i)$ $(-1,-1)$ edges. Hence, for one $(+1,+1)$ edge we get $(j-i)$ edges in $\rho$. Since, $(j - i) \geq 1$ (equality occurs when $i = (j-1)$) we get that the number of $(+1,+1)$ edges is less than or equal to the number of $(-1,-1)$ edges in $\rho$. Hence, condition (a) is not satisfied in $M \oplus N$. 

\noindent
\underline{Condition (b)}: \underline{CASE 1}: Alternating path $\rho$ starts from an unmatched man $m$: Since $m$ is unmatched in $M$ level of $m$ is 0 due to Condition $4$ of theorem \ref{thm:mPFM}. Hence $\rho$ starts with an edge present in $N$. If $\rho$ ends in a man $m'$ then $m'$ is unmatched in $N$ and hence $m'$ is a non-critical man and hence is at level 0. Let $j$ be the highest level of a man $m^j$ present in $\rho$. Since $\rho$ starts from an unmatched man $m$ which is at level 0, hence due to Condition $1$ we get that the alternating sub path of $\rho$ from $m$ to $m^j$ can contain at most $j$ $(+1,+1)$ edges. Again due to conditions $2$ and 3 we get the alternating sub path of $\rho$ from $m^j$ to $m'$ must contain $j$ $(-1,-1)$ edges. Hence $\rho$ has more $(-1,-1)$ edges than $(+1,+1)$ edges. Now, if $\rho$ ends in a woman $w'$ then $w'$ is unmatched in $M$ and hence the level of $w'$ is 0. So, $\rho$ starts at a level 0 man $m$ and ends at a level 0 woman $w'$. So, arguing similarly as we argued when $\rho$ ends in $m'$ we get that $\rho$ has more $(-1,-1)$ edges than $(+1,+1)$. \underline{CASE $2$}: Since $w$ is unmatched in $M$, hence $\rho$ starts with an edge in $N$. Alternating path $\rho$ starts from an unmatched woman $w$: If $\rho$ ends in a man $m''$ then $m''$ is unmatched in $M$ and hence due to CASE 1 we get that $\rho$ has more $(-1,-1)$ edges than $(+1,+1)$ edges. When $\rho$ ends in a woman $w''$ then $w''$ is at unmatched in $N$. If the level of $w''$ is $i$ then due to conditions $2$ and 3 we get that $\rho$ has $i$ more $(-1,-1)$ edges than $(+1,+1)$ edges. Hence, condition (b)  is not satisfied.    

\noindent\underline{Condition (c)}: Consider an alternating path $\rho$ which starts from a man $m$ matched in $M$ and ends in a woman $w$ matched in $M$. Since $m$ is the endpoint of $\rho$ we get that $m$ is unmatched in $N$. Hence, $m$ is a non-critical man and thus is at level 0. Let $w$ is at level $i$ and $\rho = (m,w_1,m_1,w_2,m_2,....,w)$. Since, $M(m) = w_1$ hence $w_1$ is at level 0. Now, from the conditions $2$ and 3 we get that $m_1$ is either at level 0 or at level 1. So, the alternating path $\rho$ can go up by only one level (that is from a woman at level $i$ to a man at level $i+1$) and if it goes up then it has to take a $(-1,-1)$ edge. Since $w$ is at level $i$ and $w_1$ is at level 0, $\rho$ will have $i$ more $(-1,-1)$ edges than $(+1,+1)$ because to go from $w_1$ to $w$ $\rho$ needs to take $i$ $(-1,-1)$ edges. Hence, condition (c) is not satisfied.  

Since none of the above conditions are satisfied, it shows that $M$ is a \PFM\ .

\textbf{$\mathbf{M}$ is a \mPFM}: Now we show that for any feasible matching $N$ such that $|N| < |M|$ we have $\phi(N,M) < \phi(M,N)$. We take the graph $M \oplus N$, which is the disjoint union of alternating paths and cycles. There is no alternating path or cycle $\rho$ in $M \oplus N$ such that $\phi((M \oplus \rho),M) > \phi(M,(M \oplus \rho))$ otherwise $M$ is not a \PFM\. So, now we need to show an alternating path or cycle in $M \oplus N$ such that $\phi((M \oplus \rho),M) < \phi(M,(M \oplus \rho))$ then only we can say $\phi(N,M) < \phi(M,N)$. Now, since $|N| < |M|$  there must exist an alternating path which starts from a man $m$ unmatched in $N$ and ends in a woman $w$ unmatched in $N$. Since $m$ is unmatched in $N$ it is a non-critical man and hence has level 0. Suppose $\rho = (m,w_1,m_1,w_2,m_2...,w)$ and the level of $w$ be $i$. Let $j$ be the highest level of a man present in $\rho$. Note that the edges $(w_i,m_i)$ are all edges present in $N$. Since $m$ is at level 0, hence $M(m) = w_1$ is also at level 0. Now from Condition $3$ of theorem \ref{thm:mPFM} we get that $m_1$ can be either at level 0 or at level 1. Hence, the alternating path $\rho$ can go up by only one level. Condition $2$ of theorem \ref{thm:mPFM} says that all the edges from a woman at level $i$ to a man at level $i+1$ is a $(-1,-1)$ edge. Since the highest level of a man in $\rho$ is $j$, hence in $\rho$ there must be $j$ $(-1,-1)$  edges. Since $\rho$ ends in a woman $w$ which is at level $i$, hence there can be at most $j - i$ $(+1,+1)$ edges as from Condition $1$ of theorem \ref{thm:mPFM} we get that $(+1,+1)$ edges are only present in between a higher level woman and a lower level man. Since $j \geq (j-i)$, hence the number of $(-1,-1)$ edges is greater than or equal to the number of $(+1,+1)$ edges in $\rho$. Hence,  $\phi((M \oplus \rho),M) < \phi(M,(M \oplus \rho))$. Note that even if the number of $(-1,-1)$ edges equal to the number of $(+1,+1)$ edges in $\rho$, we have $\phi((M \oplus \rho),M) < \phi(M,(M \oplus \rho))$ because $M \oplus \rho$ loses the votes of $m$ and $w$ (the end vertices) and does not get any extra vote from the intermediate vertices as the number of $(-1,-1)$ edges equal to the number of $(+1,+1)$ edges. Hence, $M$ is a \mPFM.  

Now we show that any $M$ that is an image of a stable matching $M'$ in $G'$ satisfies all the four conditions.

%------------------Proving four conditions------------------------------------
\underline{Condition $1$}: Suppose there is $(+1,+1)$ edge in between a man $m$ at level $i$ and woman $w$ at level $j$ such that $j \leq i$ in the matching $M$. Hence $m$ prefers $w$ more than his matched partner in $M$. Now, $M'(m^i) = M(m)$ and since the preference list of $m^i$ in the \SMone\ instance is same as the preference list of $m$ in the \SMPM\ (except the dummy women in the beginning and end of the preference list of $m^i$), $m^i$ prefers $w$ more than $M'(m^i)$. So, in $M'$ the edge $(m^i,w)$ will be a $(+1,+1)$ edge because $m^i$ prefer $w$ more than $M'(m^i)$ and $w$ prefers $m^i$ more than $M'(w)$ because her matched partner is at level $j$ and $j \leq i$. In the \SMone\ instance $w$ prefers a level $i$ man more than a level $j$ man if $i > j$ and if $i = j$ then $w$ prefers $m^i$ more than $M'(w)$ because $w$ prefers $m$ more than $M(w)$ in the matching $M$. This contradicts the fact that $M'$ is stable matching. Hence, $M$ satisfies Condition $1$.\\
%-----------------------------------------------------------------------------
\underline{Condition $2$}: Suppose there is a man $m$ at level $i$ which is adjacent to a woman at level $(i-1)$ but the edge $(m,w)$ is not labelled $(-1,-1)$. $(m,w)$ cannot be labelled $(+1,+1)$ due to Condition $1$. So, it has to be labelled $(+1,-1)$ and $(-1,+1)$. CASE 1: If $(m,w)$ is labelled $(+1,-1)$ then $m$ prefers $w$ more than $M(m)$. Hence $m^i$ prefers $w$ more than $M'(m^i)$ and $w$ prefers $m^i$ more than its matched partner in $M'$ which is the $(i-1)$ level copy of $M(w)$. Hence the edge $(m^i,w)$ is a $(+1,+1)$ edge in the matching $M'$. This contradicts stability of $M'$. CASE $2$: Now, if $(m,w)$ is labelled $(-1,+1)$ then $w$ prefers $m$ more than $M(w)$. Now since $m$ is at level $i$ so $m^i$ gets matched to a non dummy woman in the matching $M'$. So, from Corollary \ref{3.1.2m} we get that $m^{i-1}$ is matched to the dummy woman $d_m^i$ which is present at the end of his preference list. In this the edge $(m^{i-1},w)$ would be labelled $(+1,+1)$ because $m^{i-1}$ would prefer $w$ more than its matched partner in $M'$ which is present at the last of his preference list and $w$ would prefer $m^{i-1}$ more than $M'(w)$, which is a $(i-1)$ level copy of $M(w)$ as $w$ prefers $m$ more than $M(w)$. This again contradicts that $M'$ is a stable matching. Hence $M$ satisfies condition $2$.\\
%-----------------------------------------------------------------------------
\underline{Condition $3$}: Suppose Condition $3$ is not satisfied, then there is a man $m$, which at level $i$ is adjacent to a woman $w$ at level $j$ such that $j \leq (i-2)$. In this case the edge $(m^{i-1},w)$ would be a $(+1,+1)$ edge because $m^{i-1}$ prefers $w$ over its matched partner in $M'$ which is $d_m^i$ (Corollary \ref{3.1.2m}) and $w$ prefers $m^{i-1}$ over $M'(w)$ which is a $(i-2)$ level copy of $M(w)$. This contradicts the fact that $M'$ is a stable matching. Hence, $M$ satisfies Condition $3$.\\
%-----------------------------------------------------------------------------
\underline{Condition $4$}: Since, $M$ is feasible matching, so the unmatched men are only the non critical men. They must be at level 0 because there is no other copy of non critical men in $G'$. Hence $M$ satisfies Condition $4$.
\\\\
Hence any matching $M$ that is an image of a stable matching $M'$ in $G'$ is a \mPFM. \\
%--------------------------End of proving four conditions---------------------

\end{proof}
%-----------------------------------------------------------------------------

%--------------------------Surjectivity proof, moved from Section 3-----------------------
\subsection{Surjectivity of the map}
In this section, the goal is to prove the following theorem:
\begin{theorem}\label{thm:mPFM-surjectivity}
For every \mPFM\ $M$ in $G$, there exists a stable matching $M'$ in $G'$ such that $M$ is the image of $M'$.
\end{theorem}

To show the surjectivity i.e. the fact that every \mPFM\ $M$ in $G$
has a stable matching $M'$ in $G'$ as its pre-image, we first assign levels to nodes in $G$ with respect to $M$. From the assignment of levels to nodes in $G$, the pre-image $M'$ is then immediate. The assignment of levels is described in Algorithm~\ref{LAmPFM}. In the pseudocode for Algorithm~\ref{LAmPFM}, we denote the level of a vertex $v$ by $level(v)$, and the matched partner of $v$ in $M$ as $M(v)$. The proof of Theorem~\ref{thm:mPFM-surjectivity} is immediate from the correctness of Algorithm~\ref{LAmPFM}, proved below. 
%-----------------------------Levelling Algorithm for mPFM----------------------------------------------------

\begin{algorithm}[!h]
    \hspace*{\algorithmicindent}\textbf{Input:} A marriage instance $G$, set of critical nodes $C\subseteq A$, a \mPFM\ $M$ in $G$
    \\
    \hspace*{\algorithmicindent}\textbf{Output:} Assignment of levels to the vertices in $G$ based on the matching $M$. 
    \begin{algorithmic}[1]
    \State {Initially all the men and the women are assigned level $0$}
    \State flag = true
    \While{flag = true}
        \State check1 $= 0$, check2 $= 0$, check3 $= 0$
        \While{$\exists$ $m\in A,w\in B$ s.t. $level(m)=i$, $level(w)=j$, $j \leq i$, and $(m,w)$ is a $(+1,+1)$ edge}
        \State Set $level(w)=level(M(w))=i+1$ %and its matched partner $M(w)$ from level $j$ to level $(i+1)$. 
        \Comment Note that $w$ cannot be unmatched in $M$ because then $M\setminus (m,M(m))\cup (m,w)$ is more popular than $M$ and hence $M$ would not be a \PFM.
        \State check1 $= 1$
        \EndWhile
        \While{$\exists$ $m\in A,w\in B$, s.t. $level(m)=i$, $level(w)=j$, $j < i$ and $(m,w)$ is a $(+1,-1)$ or a $(-1,+1)$ edge}
        \State Set $level(w)=level(M(w))=i$ %and its matched partner $M(w)$ from level $j$ to level $i$. 
        \Comment Note that $w$ cannot be unmatched in $M$ because then $M$ would not be a \PFM.
        \State check2 $= 1$
        \EndWhile
        \While {$\exists m\in A, w\in B$ s.t. $level(m)=i$, $level(w)=j$, $j \leq (i-2)$ and $(m,w)$ is a $(-1,-1)$ edge}
        \State Set $level(w)=level(M(w))=i-1$ %and its matched partner from level $j$ to level $(i-1)$. 
        \Comment Note that $w$ cannot be unmatched in $M$ because then $M$ would not be a PFM.
        \State check3 $= 1$
        \EndWhile
        \If {check1 $= 0$ and check2 $= 0$ and check3 $= 0$}
        \State flag = false
        \EndIf
    \EndWhile
    \end{algorithmic}\caption{Leveling Algorithm for \mPFM}\label{LAmPFM}
\end{algorithm}

%-----------------------End of levelling ALgorithm for mPFM----------------- %---------------------Proving the termination of levelling Algorithm starts--------
In Algorithm~\ref{LAmPFM}, the Boolean variables check1, check2 and check3 are used to check whether the assignment of levels at any point violates one of the conditions of Theorem~\ref{thm:mPFM}.
If not, then we set flag to false and the algorithm terminates. %So, we need to prove that Algorithm~\ref{LAmPFM} terminates. 
In Theorem~\ref{thm:mPFM-term} below, we show that no level is empty. Since level of a vertex never reduces during the execution of Algorithm~\ref{LAmPFM}, it implies that the algorithm terminates.

%--------------------------Theorem 10 --------------------------------------------
\begin{theorem}\label{thm:mPFM-term}
For a man $m$ at level $i$ there exists $(i)$ either a woman $w$ at each level $j$, where $j < i$, or $(ii)$ an unmatched man $m_0$, if $j=0$ such that there is an alternating path from $w$ to $m$ or from $m_0$ to $m$ which consists of $(i-j)$ more $(+1,+1)$ edges than $(-1,-1)$ edges. 
\end{theorem}
First we show that the image $M$ of a stable matching $M'$ in $G'$ is a matching in $G$.
We use the following observation:
\begin{observation}
Every dummy woman is matched in any stable matching of $G'$. This is because each dummy woman $d_m^j$ is the first choice of $m^j$. So if $d_m^j$ is unmatched in a matching $N'$ of $G'$, then $(d_m^j,m^j)$ forms a $(+1,+1)$ edge, contradicting the stability of $N'$.
\end{observation}
\begin{lemma}\label{thm:mPFM-matching}
In any stable matching $M'$ in $G'$, at most one copy of any $m \in A$ gets matched to a non-dummy woman.  
\end{lemma}
\begin{proof}
Suppose $m^i$ be the copy of the man $m \in A$ which gets matched to a non-dummy woman. Then by the observation above, and by the fact that a dummy woman $d_m^j$ has only $m^{j-1}$ and $m^j$ in her preference list, $d_m^{i+1}$ must be matched to $m^{i+1}$, and inductively, each $d_m^j$, $j>i$ must be matched to $m^j$. 
\end{proof}

Lemma \ref{thm:mPFM-matching} shows that $M$ is a matching in $G$. 
%-----------------------------------------------------------------------------

%--------------------------Surjectivity proof for mPFM ----------------------
\begin{proof}[Proof of Theorem~\ref{thm:mPFM-term}]
We prove the statement using induction on the number of iterations of the outer while loop (line $3$). 
We refer to the three inner while loops i.e. Steps $5$ to $7$, Steps $8$ to $10$, and Steps $11$ to $13$ as Phase 1, Phase $2$, and Phase 3 of
an iteration of the outer while loop respectively. In the remainder of the proof, iteration always refers to an iteration of the outer while loop, unless stated otherwise.

{\em Base case: }
\begin{enumerate}
\item {\em The statement holds after Phase 1 of the $1$st iteration:}

In the Phase 1 of the first iteration, a man $m_i$ is assigned level $i$ only when its matched partner $M(m_i) = w_i$ has a $(+1,+1)$ edge to a man $m_{i-1}$ at level $(i-1)$. Again $m_{i-1}$ is at level $(i-1)$ after Phase 1 of the $1$st iteration because his matched partner $M(m_{i-1}) = w_{i-1}$ has a $(+1,+1)$ edge to a man $m_{i-2}$ at level $(i-2)$. Continuing this way, we get an alternating path from $m_i$ either to a woman $w_j$ at level $j$, where $j < i$, 
which has $(i-j)$ more $(+1,+1)$ edges than $(-1,-1)$ edges or to an unmatched man $m_0$ at level $0$ which has $i$ more $(+1,+1)$ than $(-1,-1)$ edges. Hence, the statement holds after Phase 1 of the $1$st iteration. 

\item {\em The statement holds after Phase $2$ of the $1$st iteration}: A man $m$ gets promoted to a level $i$ in Phase $2$ from a level $j$ where $j < i$ because his matched partner $M(m) = w$ has a $(+1,-1)$ or a $(-1,+1)$ edge to a man $m_i$ at level $i$. Let $m$ be the first man among all the men who got promoted in Phase $2$. So, $m_i$ got promoted to level $i$ in Phase 1 and thus it has an alternating path from a woman $w_j$ at level $j$, where $j < i$, or to a unmatched man $m_0$ at level $j = 0$.  

Note that $w_j \neq w$ because, in that case, the alternating path from $w_j$ to $m_i$ concatenated with the $(+1,-1)$ or $(-1,+1)$ edge $(m_i,w)$ forms an alternating cycle $\rho$ with $(i-j)$ more $(+1,+1)$ edges than $(-1,-1)$ edges, and thus $M \oplus \rho$ would become a more popular feasible matching than $M$. Now, this alternating path from $w_j$ to $m_i$ or from $m_0$ to $m_i$ concatenated with the path $(m_i,w,m)$ forms an alternating path from $w_j$ to $m$ or from $m_0$ to $m$ which has $(i-j)$ more $(+1,+1)$ edges than $(-1,-1)$ edges. Hence, there exists an alternating path from $w_j$ or from $m_0$ at level $j$ or at level $0$ respectively to the man $m$ at level $i$ with $(i-j)$ more $(+1,+1)$ edges than $(-1,-1)$ edges. Now, if $m$ is not the first man to get promoted to level $i$ during Phase $2$ then it might happen that $m$ gets promoted to level $i$ because his matched partner $M(m) = w$ has a $(+1,-1)$ edge or a $(-1,+1)$ to a man $m'$ at level $i$ who got promoted to level $i$ before $m$ during Phase $2$. In this case too there is an alternating path from a woman $w_j$ at level $j$ or from unmatched man $m_0$ at level $0$ to $m'$ which has $(i-j)$ more $(+1,+1)$ edges than $(-1,-1)$ edges (note here also $w_j \neq w$ due to the same reason). This alternating path concatenated with the path $(m',w,m)$ will give an alternating path from a woman $w_j$ at level $j$ or from $m_0$ at level $j = 0$ to $m$ which has $(i-j)$ more $(+1,+1)$ edges than $(-1,-1)$ edges. Hence, there exists an alternating path from $w_j$ at level $j$ or from $m_0$ at level $0$ to the man $m$ at level $i$ with $(i-j)$ more $(+1,+1)$ edges than $(-1,-1)$ edges. Hence, S is true after the Phase $2$ of the 1st iteration.

\item {\em The statement is true after Phase 3 of the 1st iteration: } A man $m$ gets promoted to level $i$ in Phase 3 from level $j$ where $j < i$ because his matched partner $M(m) = w$ has a $(-1,-1)$ edge to a man $m_{i+1}$ at level $(i+1)$. Let $m$ be the first man among all the men who got promotions in Phase 3. So, $m_{i+1}$ got promoted to level $i$ either in Phase 1 or in Phase $2$ and thus it has an alternating path from a woman $w_j$ at level $j$ where $j < i$ or from an unmatched man $m_0$ at level $0$ which has $i+1-j$ more $(+1,+1)$ edges than $(-1,-1)$ edges. Note that $w_j \neq w$ because in that case the alternating path from $w_j$ to $m_{i+1}$ concatenated with the $(-1,-1)$ edge $(m_{i+1},w)$ forms an alternating cycle $\rho$ with $(i-j)$ more $(+1,+1)$ edges than $(-1,-1)$ edges and thus $M \oplus \rho$ would become a more popular feasible matching than $M$. Now, this alternating path from $w_j$ to $m_{i+1}$ or from $m_0$ to $m_{i+1}$ concatenated with the path $(m_{i+1},w,m)$ forms an alternating path from $w_j$ to $m$ or from $m_0$ to $m$ which has $(i-j)$ more $(+1,+1)$ edges than $(-1,-1)$ edges because the edge $(m_{i+1},w)$ is a $(-1,-1)$ edge. Hence, there exists an alternating path from $w_j$ at level $j$ or from an unmatched man $m_0$ at level $j = 0$ to the man $m$ at level $i$ with $(i-j)$ more $(+1,+1)$ edges than $(-1,-1)$ edges. Now, if $m$ is not the first man to get promoted to level $i$ during Phase 3 then it might happen that $m$ gets promoted to level $i$ because his matched partner $M(m) = w$ has a $(-1,-1)$ edge to a man $m'$ at level $(i+1)$ who got promoted to level $(i+1)$ before $m$ during Phase 3. In this case too there is an alternating path from a woman $w_j$ at level $j, j < i$ or from an unmatched man $m_0$ at level $0$ to $m'$ which has $(i+1-j)$ more $(+1,+1)$ edges than $(-1,-1)$ edges (note that $(w_j \neq w)$ due to same reason) . This alternating path concatenated with the path $(m',w,m)$ will give an alternating path from a woman $w_j$ at level $j,j < i$ or from an unmatched man $m_0$ at level $j = 0$ to $m$ which has $(i-j)$ more $(+1,+1)$ edges than $(-1,-1)$ edges because the edge $(m',w)$ is a $(-1,-1)$ edge. Hence, there exists an alternating path from $w_j$ at level $j$ to the man $m$ at level $i$ with $(i-j)$ more $(+1,+1)$ edges than $(-1,-1)$ edges. Hence, S is true after the Phase 3 of the 1st iteration.
\end{enumerate}
Now, since the statement holds after all the  three Phases of the first iteration, it holds at the end of the first iteration.

{\em Inductive step: }Suppose the theorem statement is true after the $j^{th}$ iteration for all $j \leq k$. We prove below 
that it holds after the $(k+1)^{th}$ iteration.

{\em The statement holds after Phase 1 of the $(k+1)^{th}$ iteration}: Now, in the phase 1 of the $(k+1)^{th}$ iteration, a man $m_i$ is assigned level $i$ only when its matched partner $M(m_i) = w_i$ has a $(+1,+1)$ edge to a man $m^{i-1}$ at level $(i-1)$. Now, suppose $m_i$ be the first man who gets a promotion during Phase 1 of the $(k+1)^{th}$ iteration. So $m^{i-1}$ gets promoted to level $(i-1)$ in the previous iterations. Hence, due to inductive hypothesis we get that there exists a woman $w_j$ at level $j$ where $j<i$ or from an unmatched man $m_0$ at level $j = 0$ such that there is an alternating path from $w_j$ or from $m_0$ to $m^{i-1}$ with $(i-1-j)$ more $(+1,+1)$ edges than $(-1,-1)$ edges. Now, concatenating this path with $(m^{i-1},w_i,m_i)$ we get that there is an alternating path either from $w_j$ or from $m_0$ to $m_i$ with $(i-j)$ more $(+1,+1)$ edges than $(-1,-1)$ edges. Now, if $m_i$ is not the first man who gets a promotion during Phase 1 of the $(k+1)^{th}$ iteration then it might happen that $m^{i-1}$ gets promoted to level $(i-1)$ before $m_i$ during Phase 1 of $(k+1)^{th}$ iteration. In this case also  there exists a woman $w_j$ at level $j$ where $j<i$ or an unmatched man $m_0$ at level $j = 0$ such that there is an alternating path either from $w_j$ or from $m_0$ to $m^{i-1}$ with $(i-1-j)$ more $(+1,+1)$ edges than $(-1,-1)$ edges. Now, concatenating this path with $(m^{i-1},w_i,m_i)$ we get that there is an alternating path either from $w_j$ or from $m_0$ to $m_i$ with $(i-j)$ more $(+1,+1)$ edges than $(-1,-1)$ edges. Hence, S is true after Phase 1 of the $(k+1)^{th}$ iteration.

For the remaining two phases, the proof is similar to the respective proofs of the 1st iteration. 
This completes the proof of the theorem.
\end{proof}

\begin{theorem}\label{thm:mPFM-stability}
The matching $M'$ is a stable matching in $G'$.
\end{theorem}
\begin{proof}
Suppose $M'$ is not a stable matching. Then there exists a pair $(a,b)$ such that $(a,b)$ is a $(+1,+1)$ edge where $a \in A'$ and $b \in B'$. 

{\em Case 1: The woman $b$ is a dummy woman:} Let $b = d_m^i$. Now, the only two neighbors of $d_m^i$ are $m^i$ and $m^{i-1}$. If $M'(d_m^i)=m^{i-1}$ then 
$(m^i,d_m^i)$ cannot be a $(+1,+1)$ edge because $d_m^i$ prefers $m^{i-1}$ the most. If $M'(d_m^i)=m^i$ then $(m^{i-1},d_m^i)$ cannot be a $(+1,+1)$ edge because $m^{i-1}$ prefers $d_m^i$ the least and %since $d_m^i$ is matched with $m^i$ in $M'$ then 
by our mapping of $M$ to $M'$, $M'(m^{i-1})$ is either $d_m^{i+1}$ or a non-dummy woman. So a dummy woman cannot participate in a blocking pair with respect to $M'$. 

{\em Case $2$: The woman $b$ is a non-dummy woman:} Let $M'(b)=m^i$ %i.e. a man at level $i$. 
By the execution of Algorithm~\ref{LAmPFM}, $b$ does not
have an edge to a man at level $i+2$ or higher. So the man $a$ cannot be at level $i+2$ or higher. Moreover, $a$ cannot be at a level $j<i$ since, in
$G'$, $b$ prefers any man at level $i$ over any man at a level $j<i$. If $a$ is at level $i$, then the edge $(a,b)$ is also a $(+1,+1)$ edge with respect to $M$
in $G$. But then Algorithm~\ref{LAmPFM} would have moved $b$ to a higher level. So $a$ cannot be at level $i$. If $a$ is at level $i+1$, then the 
edge $(a,b)$ is a $(-1,-1)$ edge in $G$ by the execution of Algorithm~\ref{LAmPFM}. Since $a$ has the same preference list except possibly for the
addition of dummy women, $a$ does not prefer $b$ over $M(a)$ and hence over $M'(a)$. So $(a,b)$ cannot be a blocking pair with respect to $M'$.

Since no woman can participate in a blocking pair with respect to $M'$, stability of $M'$ follows.
\end{proof}

%---------------Termination of leveling algorithm for mPFM completed------------------

The following corollary is a straight forward consequence of Theorem~\ref{thm:mPFM-term} and the fact that $M$ is a \mPFM.

%----------------------------------Corollary 11 ---------------------------------
%\begin{theorem}
\begin{corollary}\label{cor:levels}
All non-critical men are assigned level zero and the critical men are assigned level less than or equal to $|C|$ by Algorithm~\ref{LAmPFM}.
\end{corollary}
\begin{proof}
Suppose there is a non-critical man $m_i$ at level $i, i > 0$. Now, from Theorem \ref{thm:mPFM-term}, there is an alternating path, say $\rho$, from a woman $w_0$ at level $0$ or from an unmatched man $m_0$ to $m_i$ which has $i$ more $(+1,+1)$ edges than $(-1,-1)$ edges. Let $N = M \oplus \rho$. Observe that Algorithm~\ref{LAmPFM} assigns level $0$ to all the unmatched men in $M$, so $m_i$ is matched. Now, it is easy to see that $N$ is also a \PFM. But $|N|<|M|$, so this contradicts the assumption that $M$ is a \mPFM.
After assigning levels to the vertices in $G$, the pre-image of $M$ i.e. a stable matching $M'$ in $G'$ is constructed as follows. If a man $m$ in $G$ gets assigned level $i$ then $M'(m_i) = M(m)$. If $m$ is unmatched in $M$, then $m\notin C$ by feasibility of $M$, and $m$ gets level $0$ by Corollary~\ref{cor:levels}. Then we leave $m$ unmatched in $M'$ as well. For $j < i$, $m^j$ gets matched to the dummy woman $d_m^{j+1}$ and for $j > i$, $m^j$ gets matched to the dummy woman $d_m^j$. %In Theorem~\ref{thm:mPFM-stability} in Appendix, we prove that $M'$ is stable in $G'$. %The proof is given in Appendix.
\end{proof}
%-------------------------Theorem 21 ----------------------------------------------

The reduction and proofs for \DFM\ are similar, and are given in Appendix for the sake of completeness. 

%% file: 3-Algorithm.tex
\section{The Popular Edge Algorithm}
\label{sec:main-thm}
Now we are ready to prove Theorem~\ref{thm:main}, from which, the popular edge algorithm is as follows.
%-----------------------------------------------------------------------------
%	\subsection{High Level idea}
%The high level idea to solve the Popular Edge problem in a \SMPM\  instance goes like this. We prove in Theorem \ref{thm:main} that if an edge $e \in E$ belongs to a \PFM\ in a given \SMPM\  instance $G$ then $e$ belongs to either a \mPFM\  or to a \DFM\  in $G$. 
For a given edge $e=(m,w)$, we construct $G',G''$ using reductions from Section~\ref{sec:algo-mPFM-DFM} and check if any of the edges $(m^i,w)$ in $G'$ or $G''$ belong to a stable matching in the respective instance using Knuth's algorithm for stable edges \cite{Knuth76}.

If there is a \mPFM\ or a \DFM\ containing $e$ then there is nothing to prove. So we need to prove the theorem for an edge $e$ that is contained
in a \PFM\ $M$ that is neither a \mPFM\ nor a \DFM, and show that there is also a \mPFM\ or a \DFM\ containing $e$.
The proof of Theorem~\ref{thm:main} involves the following two results:
\begin{theorem}\label{thm:partition}
If $M$ is neither a \mPFM\ or a \DFM, then $A\cup B$ can be partitioned into three parts $A_d\cup B_d$, $A_m\cup B_m$ and $A_r\cup B_r$ such that no edge of $M$ is present in $A_i \times B_j$, $i \neq j,$  where $i,j \in \{d, m, r\}$.
\end{theorem}
We prove Theorem~\ref{thm:partition} in Section~\ref{sec:part}.
Because of Theorem~\ref{thm:partition}, it follows that the partition of $A\cup B$ also induces a partition of $M$ into three parts, say $M_d,M_m,M_r$ respectively. The following theorem shows that either $M_d$ or $M_m$ can be transformed into another matching so that
the resulting matching is a \mPFM\ or a \DFM\ in $G$. 
\begin{theorem}\label{thm:convert}
There exist algorithms to transform:
\begin{enumerate}
\item the matching $M_d$ to another matching $M'_m$ in $A_d\cup B_d$ such that $M^*_m=M'_m\dot{\cup} M_m\dot{\cup} M_r$ is a \mPFM\ in $G$ 
\item the matching $M_m$ to another matching $M'_d$ in $A_m\cup B_m$ such that $M^*_d=M_d\dot{\cup} M'_d\dot{\cup} M_r$ is a \DFM\ in $G$.
\end{enumerate}
\end{theorem}
A proof of Theorem~\ref{thm:convert} is given in Section~\ref{sec:convert}. From Theorems \ref{thm:partition} and \ref{thm:convert}, Theorem~\ref{thm:main} follows:
\begin{proof}[Proof of Theorem~\ref{thm:main}]
Depending on the part that contains the given edge $e$, one of the two transformations mentioned in Theorem~\ref{thm:convert} can be applied:
If $e \in M_m$ (respectively $e\in M_d$), apply the first (respectively, second) transformation from Theorem~\ref{thm:convert} i.e. convert the matching $M_d$ to $M'_m$ ($M_m$ to $M'_d$). Then, by Theorem~\ref{thm:convert}, the resulting matching $M^*_m$ ($M^*_d$) is a \mPFM\ (\DFM) in $G$ containing $e$. If $e \in M_r$, we can apply any one of the two transformations mentioned in Theorem~\ref{thm:convert}.
Thus, in all the three cases, we get a \mPFM\ or a \DFM\ containing $e$. This completes the proof of Theorem~\ref{thm:main}.
\end{proof}
%\begin{figure}
%    \centering
%    \includegraphics[width=12cm, height=5cm]{Partition.png}
%    \caption{Case 1 (Green colour): $e \in M_m$, in this case we convert $M_d$ to $M_m'$ and then we show $M'$ = $(M'_m \cup M_m \cup M_r)$ is a \mPFM\ . Case 2 (Red Colour): $e \in M_d$, in this case we convert $M_m$ to $M_d^*$ and then we show $M^*$ = $(M_d \cup M_d^* \cup M_r)$ is a \DFM\ . Case 3 (Black Colour): $e \in M_r$, in this case either we convert $M_m$ to $M_d^*$ or $M_d$ to $M_m'$. }
%    \label{fig:part}
%\end{figure}
%---------------------------------------------------------------------------------------------------------------------------------------------------
\subsection{Partition Method}\label{sec:part}
We prove Theorem~\ref{thm:partition} now. 
For partitioning $A\cup B$ and $M$, we first assign levels to the vertices of $A\cup B$ using Algorithm~\ref{LAmPFM} described in Section~\ref{sec:mPFM-algo}. We refer to the level of a vertex $u\in A\cup B$ as $level(u)$. Since $M$ is a \PFM\ but not a \mPFM\ by assumption, all the non-critical men may
not be at level $0$. However, the following holds:
\begin{lemma}\label{lem:lvl of np man}
After applying Algorithm~\ref{LAmPFM} on a \PFM\ $M$ all non-critical men are assigned levels $0$ or $1$.
\end{lemma}

\begin{proof}
If there is a non-critical man $m$ who is assigned level $i \geq 2$, then according to Theorem~\ref{thm:mPFM-term} there exists a woman at level $0$ such that $m$ has an alternating path $\rho$ from $w$ with $i$ more $(+1,+1)$ edges than $(-1,-1)$ edges. Since $i \geq 2$, the matching $M \oplus \rho$ is feasible and is more popular than $M$, contradicting the assumption that $M$ is a \PFM\ . Hence, all non-critical men are assigned levels $0$ or $1$.  
\end{proof}
The following notions will be used in the partitioning procedure.
%------------------Definition of SRAP---------------------------------------
\begin{definition}[Size Reducing Alternating Path (\SRAP)]: An alternating path $\rho$ with respect to a matching $M$ is called as \SRAP\ if the following conditions are satisfied:
\begin{enumerate}
\item The number of $(+1,+1)$ edges in $\rho$ is one more than the number of $(-1,-1)$ edges in $\rho$,
\item It starts in a matched woman at level $0$.
\item It ends in a non-critical man at level $1$.
\end{enumerate}
\end{definition}
%---------------------------------------------------------------------------

%----------------------Definition of SIAP------------------------------------
\begin{definition}[Size Increasing Alternating Path (\SIAP)]: An alternating path $\rho$ with respect to a matching $M$ is called an \SIAP\ if the following conditions are satisfied. 
\begin{enumerate}
    \item There are an equal number of $(+1,+1)$ and $(-1,-1)$ edges in $\rho$.
    \item Its end-points are an unmatched man and an unmatched woman.
\end{enumerate}
\end{definition}
Intuitively, if $\rho$ is an \SIAP\ (respectively an \SRAP), then $M\oplus\rho$ gives a larger (respectively, smaller) \PFM.
%----------------------------------------------------------------------------
Theorem~\ref{thm:SIAP-SRAP} below shows that an \SRAP\ and an \SIAP\ must exist if $M$ is not a \mPFM\ or a \DFM. %Theorem~\ref{thm:SIAP-SRAP-disjoint} shows that an \SIAP\ and an \SRAP\ are disjoint. 

%------------------------Theorem 27----------------------------------------

\begin{theorem} \label{thm:SIAP-SRAP}%old label 2t
If a \PFM\ $M$ is neither a \mPFM\  nor a \DFM, then $G$ must contain an \SRAP\ and an \SIAP\ with respect to $M$.
\end{theorem}
\begin{proof}
Suppose $M_{min}$ is a \mPFM. Consider $M \oplus M_{min}$ which is a disjoint union of alternating paths and cycles. Since $|M_{min}| < |M|$, there exists an alternating path $\rho$ in $M \oplus M_{min}$ whose both end-points are matched in $M$. By popularity of $M$ and $M_{min}$, $\rho$ must have one more $(+1,+1)$ edge than $(-1,-1)$ edges %because only in that case we have $\phi(M \oplus \rho, M)$ = $\phi(M, M \oplus \rho)$. If $\phi(M \oplus \rho, M) > \phi(M, M \oplus \rho)$ then $M$ is not a popular matching and if $\phi(M \oplus \rho, M) < \phi(M, M \oplus \rho)$ then $M_{min}$ is not a popular matching. Hence, we must have 
so that $\phi(M \oplus \rho, M) = \phi(M, M \oplus \rho)$. By feasibility of $M_{min}$, $\rho$ must have a non-critical man $m$ as one of its end-points, since $m$ is unmatched in $M_{min}$. %because both $M$ and $M_{min}$ are feasible matchings and the end points are matched in $M$ but not in $M_{min}$. 
The level assigned to $m$ has to be $1$ because non-critical men can only be assigned levels $0$ or $1$ by Lemma \ref{lem:lvl of np man}. Moreover, 
since $\rho$ has $1$ more $(+1,+1)$ edge than $(-1,-1)$ edges,% the level assigned to $m$ is 1 and 
%The other end-point of $\rho$ is a matched woman $w$. 
the level assigned to the other end-point $w$ is $0$. Recall that $w$ is matched in $M$ and unmatched in $M_{min}$.
Hence $\rho$ is an \SRAP.

Now suppose $M_d$ is a \DFM\ in $G$. The graph $M \oplus M_d$ is a disjoint union of alternating paths and cycles. Since $|M| < |M_d|$, there must exist an alternating path $\rho$ in $M \oplus M_d$ whose end-points are unmatched in $M$ and matched in $M_d$. %So, the endpoints are unmatched in $M$. 
Here, $\rho$ must have an equal number of $(+1,+1)$ and $(-1,-1)$ edges, otherwise $(M \oplus \rho)$ becomes a more popular matching than $M$. Hence $M$ has an \SIAP.  
\end{proof}

The partitioning is based on \SIAP\ and \SRAP, so the following theorem is essential for the
partitioning to be well-defined. Theorem~\ref{thm:SIAP-SRAP-disjoint} below shows that no vertex belongs to both an \SRAP\ and an \SIAP:

\begin{theorem}\label{thm:SIAP-SRAP-disjoint}%old label 3t
Given a \PFM\ $M$, no vertex in $G$ belongs to both an \SIAP\ and an \SRAP.
\end{theorem}
\begin{proof}
Note that if a man $m$ belongs to both an \SIAP\ $\rho$ and an \SRAP\ $\sigma$, then his matched partner $M(m)$ must belong to both $\rho$ and $\sigma$ too. Also note that no man or woman unmatched in $M$ can belong to both $\rho$ and $\sigma$ because all the men and women in an \SRAP\ are matched in $M$. Suppose a matched pair $(m,w)$ in $M$ belongs to both $\rho$ and $\sigma$. Let $m_I$ and $w_I$ be the end-points of $\rho$ and $m_R$ and $w_R$ be the end-points of $\sigma$. Let the level assigned to the pair $(m,w)$ be $i_{(m,w)}$. Since $m_I$ and $w_I$ are unmatched in $M$, both of them are assigned level $0$. Now, the alternating subpath $\rho_I$ of 
$\rho$ from $w_I$ to $m$ must contain $i_{(m,w)}$ more $(-1,-1)$ edges than $(+1,+1)$ edges. This is because all the adjacent vertices of the women present in level $i$ must be present at a level $j$ where $j \leq (i+1)$. So, $\rho_I$ can go up by only one level that is from a level $i$ woman it can only go to a level $i+1$ man and, from the properties of Algorithm~\ref{LAmPFM}, we also know that all the edges between level $i$ women and level $(i+1)$ men are $(-1,-1)$ edges. Now, since an \SIAP\ has an equal number of $(+1,+1)$ and $(-1,-1)$ edges, we have that the alternating subpath $\rho_{I'}$ of $\rho$ from $m_I$ to $m$ must contain $i_{(m,w)}$ more $(+1,+1)$ edges than $(-1,-1)$ edges. 

%Now again 
By a similar argument, the alternating subpath $\sigma_R$ of $\sigma$ starting from $m$ to $m_R$ consists of $(i_{(m,w)} - 1)$ more $(-1,-1)$ edges than $(+1,+1)$ edges. %This is because of the same reasoning given for $\rho_I$. 
Thus the path $\beta = \rho_{I'}\circ \sigma_{R}$, where $\circ$ denotes concatenation, contains more $(+1,+1)$ edges than $(-1,-1)$ edges.% of two alternating paths that is we travel 
This is because $\rho_{I'}$ contains $i_{(m,w)}$ more $(+1,+1)$ edges than $(-1,-1)$ edges, and then $\rho_{R}$ has $(i_{(m,w)} - 1)$ more $(-1,-1)$ edges than $(+1,+1)$ edges. The matching $M \oplus \beta$ is more popular than $M$ because $\beta$ has one more $(+1,+1)$ edge than the number of $(-1,-1)$ edges, and $M\oplus\beta$ has the same size as that of $M$. This contradicts the fact that $M$ is a \PFM. Hence, no vertex in $G$ can belong to both an \SIAP\ and an \SRAP\ for a given matching $M$.
\end{proof}

%---------------------------------------------------------------------------

%-----------------------------------------------------------------------------

%-----------------------------Partition Method-----------------------------------------------
Now we give the method to partition $A\cup B$ below, as required by Theorem~\ref{thm:partition}. %$M$ in $G$ into three parts $M_d$, $M_m$ and  $M_r$.

\subsubsection{Partitioning $A\cup B$}
\begin{enumerate}[(a)]
    \item Initialize $A_d$, $A_m$, $A_r$, $B_d$, $B_m$, $B_r$ to empty sets.
    \item For all unmatched men $(m_u)$ and unmatched women $(w_u)$ in $M$ we do: $A_m = A_m \cup \{m_u\}$ and $B_m = B_m \cup \{w_u\}$
    \item\label{itm:SRAP} From Theorem~\ref{thm:SIAP-SRAP}, we know that $M$ must have an \SRAP\ and an \SIAP. For all men $m_d$ and women $w_d$ in each \SRAP\ do: $A_d = A_d \cup \{m_d\}$ and $B_d = B_d \cup \{w_d\}$
    \item\label{itm:SIAP} For all men $m$ and women $w$ in each \SIAP\ do: $A_m = A_m \cup \{m\}$ and $B_m = B_m \cup \{w\}$
    \item\label{itm:M_d} While there exists a $(+1,+1$) edge $(m,w)$ such that $m \in A \backslash (A_d \cup A_m)$, $level(m)=i$, $level(w)=j$, $j \leq i+1$ do: $A_d = A_d \cup \{m\}$ and $B_d = B_d \cup \{M(m)\}$
    \item\label{itm:M_m} While there exists a $(+1,+1)$ edge $(m,w)$ such that $m\in A_m$, $level(m)=i$, $w \in B \backslash (B_d \cup B_m)$, $level(w)=j$, $j \leq i+1$, %or to a woman $w \in B \backslash (B_d \cup B_m)$ at level $(i+1)$ such that $(m_m,w)$ is a $(+1,+1)$ edge, 
    do: $B_m = B_m \cup \{w\}$ and $A_m = A_m \cup \{M(w)\}$
    \item $A_r = A \backslash (A_d \cup A_m)$ and $B_r = B \backslash (B_d \cup B_m)$. Let $M_d,M_m,M_r$ be the parts of $M$ present in the induced subgraph on $A_d\cup B_d$, $A_m\cup B_m$, and $A_r\cup B_r$ respectively.
\end{enumerate}

%--------------------------------------------------------------------------
To complete the proof of Theorem~\ref{thm:partition}, we need to show that the above procedure partitions $A\cup B$ i.e., the three sets $A_d\cup B_d,A_m\cup B_m, A_r\cup B_r$ are disjoint. The partition procedure always puts a vertex and its matched partner in the same partition. So it is immediate that $M_d,M_m,M_r$ partition $M$. 

In the discussion below, we retain the same assignment of levels to all the vertices as was done before partitioning. 

We show that the sets $A_d$ and $A_m$ are disjoint. This implies that $B_d$ and $B_m$ are disjoint as well, since they consist of matched partners of the men in $A_d$ and $A_m$ respectively. 
From Theorem \ref{thm:SIAP-SRAP-disjoint}, a man $m$ cannot be a part of both an \SIAP\ and an \SRAP, and thus $m$ cannot be added in both $A_d$ and $A_m$ in the steps \ref{itm:SRAP} and \ref{itm:SIAP}. Hence, $A_d$ and $A_m$ remain disjoint in these steps.
We need to show that the three sets remain disjoint in steps \ref{itm:M_d}, and \ref{itm:M_m}. In the following, we show that an analogue
of Theorem~\ref{thm:mPFM-term} holds for the induced graphs on $A_d\cup B_d$ and $A_m\cup B_m$. 

%--------------------------------Lemma 29---------------------------------------------
\begin{lemma} \label{1c}
For a man $m \in A_m$ at level $i$, there is an alternating path with $i$ more $(+1,+1)$ edges than $(-1,-1)$ edges which starts at $m_I$ and ends in $m$ where $m_I$ is an unmatched man and also an endpoint of an \SIAP\ with respect to $M$. 
Analogously, for a woman $w \in B_d$ at level $i$ there is an alternating path with $(i-1)$ more $(-1,-1)$ edges than $(+1,+1)$ edges which starts at $m_R$ and ends in $w$ where $m_R$ is a non-critical man and also an endpoint of an \SRAP\ with respect to $M$.
\end{lemma}

%--------------------
\begin{proof}%[Proof of Lemma~\ref{1c}]
A man $m_m$ is added in $A_m$ either in step (d) or in step (f) of the Partition Method. If $m_m$ is added in step (d) then $m_m$ is a part of a SIAP $\rho$. Let the level assigned to $m_m$ is $i$. Let $m_I$ and $w_I$ be the endpoints of $\rho$. Hence, $m_I$ and $w_I$ are unmatched in $M$ and thus the levels of $m_I$ and $w_I$ are 0. Now, a woman at level $i$ is adjacent to a man at level $j$ where $j \leq (i+1)$. Let $\rho'$ be the alternating subpath of $\rho$ which starts from $w_I$ and ends in $m_m$.  So $\rho'$ can go up by only one level that is $\rho'$ can traverse from a woman at level $i$ to a man at level $(i+1)$. We know that a edge from a woman at level $i$ to a man at level $(i+1)$ is a $(-1,-1)$ edge. Hence, the alternating subpath $\rho'$ of $\rho$ which starts from $w_I$ and ends in $m_m$ consists of $i$ more $(-1,-1)$ edges than $(+1,+1)$ edges. Now, since $\rho$ consists of equal number of $(+1,+1)$ and $(-1,-1)$ edges so the alternating subpath $\rho''$ of $\rho$ starting from $m_I$ to $m_m$ consists of $i$ more $(+1,+1)$ edges than $(-1,-1)$ edges.

Let, $m_1$ is the first man who is added in $A_m$ in step (f). Then $m_1$ is added to $A_m$ because his matched partner $M(m_1) = w_1$ is adjacent to $m_m$ which is added to $A_m$ in step (d). Let the level of $m_m$ is $i$. Hence, from the arguments given in the previous paragraph we get that there is an alternating path from $m_I$ to $m_m$ which has $i$ more $(+1,+1)$ edges than $(-1,-1)$ edges. Now, if the level of $m_1$ and $w_1$ is $i$ then the edge $(m_m,w_1)$ is either a $(+1,-1)$ or a $(-1,+1)$ edge. Hence, the alternating path from $m_I$ to $m_m$ concatenated with the path $(m_m,w_1,m_1)$ is the alternating path from $m_I$ to $m_1$ which has $i$ more $(+1,+1)$ edges than $(-1,-1)$ edges, if the level of $m_1$ and $w_1$ is $(i-1)$ then the edge $(m_m,w_1)$ is a $(-1,-1)$ edge. Hence, the alternating path from $m_I$ to $m_m$ concatenated with the path $(m_m,w_1,m_1)$ is the alternating path from $m_I$ to $m_1$ which has $(i-1)$ more $(+1,+1)$ edges than $(-1,-1)$ edges and if the level of $m_1$ and $w_1$ is $(i+1)$ then the edge $(m_m,w_1)$ is a $(+1,+1)$ edge. Hence, the alternating path from $m_I$ to $m_m$ concatenated with the path $(m_m,w_1,m_1)$ is the alternating path from $m_I$ to $m_1$ which has $(i+1)$ more $(+1,+1)$ edges than $(-1,-1)$ edges. 

Now if $m_k$ be the $k^{th}$ man added to $A_m$ in step (f) (where $k \geq 1$) and for all $j \leq k$ we assume that if the level of $m_j$ is $i$ then there exists an alternating path from $m_I$ to $m_j$ which has $i$ more $(+1,+1)$ edges than $(-1,-1)$ edges. Now if $m_{k+1}$ is the $(k+1)^{th}$ man who is added in $A_m$ in step (f). Then $m_{k+1}$ is added to $A_m$ because his matched partner $M(m_{k+1}) = w_{k+1}$ is adjacent to $m_m$ which is added to $A_m$ in step (d) or in step (f). Let the level of $m_m$ is $i$. Hence, we get that there is an alternating path from $m_I$ to $m_m$ which has $i$ more $(+1,+1)$ edges than $(-1,-1)$ edges. Now, if the level of $m_{k+1}$ and $w_{k+1}$ is $i$ then the edge $(m_m,w_{k+1})$ is either a $(+1,-1)$ or a $(-1,+1)$ edge. Hence, the alternating path from $m_I$ to $m_m$ concatenated with the path $(m_m,w_{k+1},m_{k+1})$ is the alternating path from $m_I$ to $m_{k+1}$ which has $i$ more $(+1,+1)$ edges than $(-1,-1)$ edges, if the level of $m_{k+1}$ and $w_{k+1}$ is $(i-1)$ then the edge $(m_m,w_{k+1})$ is a $(-1,-1)$ edge. Hence, the alternating path from $m_I$ to $m_m$ concatenated with the path $(m_m,w_{k+1},m_{k+1})$ is the alternating path from $m_I$ to $m_{k+1}$ which has $(i-1)$ more $(+1,+1)$ edges than $(-1,-1)$ edges and if the level of $m_{k+1}$ and $w_{k+1}$ is $(i+1)$ then the edge $(m_m,w_{k+1})$ is a $(+1,+1)$ edge. Hence, the alternating path from $m_I$ to $m_m$ concatenated with the path $(m_m,w_{k+1},m_{k+1})$ is the alternating path from $m_I$ to $m_{k+1}$ which has $(i+1)$ more $(+1,+1)$ edges than $(-1,-1)$ edges.

Hence, for a man $m \in A_m$ which is at level $i$ there is an alternating path with $i$ more $(+1,+1)$ edges than $(-1,-1)$ edges which starts at $m_I$ and ends in $m$ where $m_I$ is an unmatched man and also an endpoint of a SIAP in the PFM $M$.

Now we prove the statement for a woman $w_d \in B_d$.
A woman $w_d$ is added in $B_d$ either in step (c) or in step (e) of the Partition Method. If $w_d$ is added in step (c) then $w_d$ is a part of a SRAP $\rho$. Let the level assigned to $w_d$ is $i$. Let $m_R$ and $w_R$ be the endpoints of $\rho$. Hence, $m_R$ is a non-critical man and $w_R$ is a woman matched in $M$ and the levels of $m_R$ and $w_R$ are 1 and 0 respectively (according to the definition of SRAP). Now, a woman at level $i$ is adjacent to a man at level $j$ where $j \leq (i+1)$. Let $\rho'$ be the alternating subpath of $\rho$ which starts with the matched edge $(m_R, M(m^R)$ and ends in $w_d$.  So $\rho'$ can go up by only one level that is $\rho'$ can traverse from a woman at level $i$ to a man at level $(i+1)$. We know that a edge from a woman at level $i$ to a man at level $(i+1)$ is a $(-1,-1)$ edge. Hence, the alternating subpath $\rho'$ of $\rho$ which starts from $m_R$ and ends in $w_d$ consists of $(i-1)$ more $(-1,-1)$ edges than $(+1,+1)$ edges because $m^R$ is at level 1 and $w_d$ is at $i$.

Let, $w_1$ is the first woman who is added in $B_d$ in step (e). Then $w_1$ is added to $B_d$ because her matched partner $M(w_1) = m_1$ is adjacent to $w_d$ which is added to $B_d$ in step (c). Let the level of $m_1$ is $i$. Now, if the level of $w_d$ is $i$ then there is an alternating path from $m_R$ to $w_d$ which has $(i-1)$ more $(-1,-1)$ edges than $(+1,+1)$ edges and the edge $(m_1,w_d)$ is either a $(+1,-1)$ or a $(-1,+1)$ edge. Hence, the alternating path from $m_R$ to $w_d$ concatenated with the path $(w_d,m_1,w_1)$ is the alternating path from $m_R$ to $w_1$ which has $(i-1)$ more $(-1,-1)$ edges than $(+1,+1)$ edges, if the level of $w_d$ is $(i-1)$ then there is an alternating path from $m_R$ to $w_d$ which has $(i-2)$ more $(-1,-1)$ edges than $(+1,+1)$ edges.the edge $(m_1,w_d)$ is a $(-1,-1)$ edge. Hence, the alternating path from $m_R$ to $w_d$ concatenated with the path $(w_d,m_1,w_1)$ is the alternating path from $m_R$ to $w_1$ which has $(i-1)$ more $(-1,-1)$ edges than $(+1,+1)$ edges and if the level of $w_d$ is $(i+1)$ then there is an alternating path from $m_R$ to $w_d$ which has $i$ more $(-1,-1)$ edges than $(+1,+1)$ edges and the edge $(m_1,w_d)$ is a $(+1,+1)$ edge. Hence, the alternating path from $m_R$ to $w_d$ concatenated with the path $(w_d,m_1,w_1)$ is the alternating path from $m_R$ to $w_1$ which has $(i-1)$ more $(-1,-1)$ edges than $(+1,+1)$ edges.

Now if $w_k$ be the $k^{th}$ woman added to $B_d$ in step (e) (where $k \geq 1$) and for all $j \leq k$ we assume that if the level of $w_j$ is $i$ then there exists an alternating path from $m_R$ to $w_j$ which has $(i-1)$ more $(-1,-1)$ edges than $(+1,+1)$ edges. Now if $w_{k+1}$ is the $(k+1)^{th}$ woman who is added in $B_d$ in step (e). Then $w_{k+1}$ is added to $B_d$ because her matched partner $M(w_{k+1}) = m_{k+1}$ is adjacent to $w_d$ which is added to $B_d$ in step (c) or in step (e). Let the level of $m_{k+1}$ is $i$. Now, if the level of $w_d$ is $i$ then there is an alternating path from $m_R$ to $w_d$ which has $(i-1)$ more $(-1,-1)$ edges than $(+1,+1)$ edges and the edge $(m_{k+1},w_d)$ is either a $(+1,-1)$ or a $(-1,+1)$ edge. Hence, the alternating path from $m_R$ to $w_d$ concatenated with the path $(w_d,m_{k+1},w_{k+1})$ is the alternating path from $m_R$ to $w_{k+1}$ which has $(i-1)$ more $(-1,-1)$ edges than $(+1,+1)$ edges, if the level of $w_d$ is $(i-1)$ then then there is an alternating path from $m_R$ to $w_d$ which has $(i-2)$ more $(-1,-1)$ edges than $(+1,+1)$ edges and the edge $(m_{k+1},w_d)$ is a $(-1,-1)$ edge. Hence, the alternating path from $m_R$ to $w_d$ concatenated with the path $(w_d,m_{k+1},w_{k+1})$ is the alternating path from $m_R$ to $w_{k+1}$ which has $(i-1)$ more $(-1,-1)$ edges than $(+1,+1)$ edges and if the level of $w_d$ is $(i+1)$ then there is an alternating path from $m_R$ to $w_d$ which has $i$ more $(-1,-1)$ edges than $(+1,+1)$ edges and the edge $(m_{k+1},w_d)$ is a $(+1,+1)$ edge. Hence, the alternating path from $m_R$ to $w_d$ concatenated with the path $(w_d,m_{k+1},w_{k+1})$ is the alternating path from $m_R$ to $w_{k+1}$ which has $(i+1)$ more $(+1,+1)$ edges than $(-1,-1)$ edges.

Hence, for a woman $w \in B_d$ which is at level $i$ there is an alternating path with $(i-1)$ more $(-1,-1)$ edges than $(+1,+1)$ edges which starts at $m_R$ and ends in $w$ where $m_R$ is an unmatched man and also an endpoint of a SRAP in the PFM $M$.
\end{proof}

\begin{lemma} \label{3l}
For an edge $(m,w) \in A_m \times B_d$ in $G$ we have the following
\begin{enumerate}[(i)]
    \item If $m$ is at level $i$ and $w$ is at level $(i+1)$ then the edge $(m,w)$ is not a $(+1,+1)$ edge.
    \item If $m$ is at level $i$ then $w$ cannot be at level $(i-1)$ or below.
    \item If $m$ is at level $i$ and $w$ is at level $i$ then $(m,w)$ is a $(-1,-1)$ edge. 
\end{enumerate}
\end{lemma}
\begin{proof}

\textbf{Condition (i)}: Suppose such a pair $(m,w)$ exists. %there is a man $m \in A_m$ at level $i$ which is adjacent to a woman $w \in B_d$ at level $(i+1)$ such that $(m,w)$ is a $(+1,+1)$ edge. So, 
From Lemma \ref{1c}, there is an alternating path $\rho_I$ from $m$ to $m_I$ with $i$ more $(+1,+1)$ edges than $(-1,-1)$ edges, where $m_I$ is an unmatched man and also an endpoint of an \SIAP, and there is an alternating path $\rho_R$ from $w$ to $m_R$ with $i$ more $(-1,-1)$ edges than $(+1,+1)$ edges, where $m_R\in A\setminus P$ and $m_R$ is an endpoint of an \SRAP. Hence the alternating path $\rho$ = $\rho_I \circ\rho_R$ has more $(+1,+1)$ edges than $(-1,-1)$ edges. Here $\circ$ denotes concatenation. 

\textbf{Condition (ii)}: If $(m,w)$ is an edge in $G$ and level of $m$ is $i$ then the level of $w$ cannot be less than $(i-1)$. 
Suppose there is a man $m \in A_m$ at level $i$ adjacent to a woman $w \in B_d$ at level $(i-1)$. Note that $(m,w)$ is a $(-1,-1)$ edge because all edges between a man at level $i$ and a woman at level $(i-1)$ are $(-1,-1)$ edges. So, from Lemma \ref{1c}, there is an alternating path $\rho_I$ from $m$ to $m_I$ (where $m_I$ is an unmatched man and also an endpoint of a SIAP) with $i$ more $(+1,+1)$ edges than $(-1,-1)$ edges. Again from Lemma \ref{1c} we get that there is an alternating path $\rho_R$ from $w$ to $m_R$ (where $m_R$ is a non-critical man and also an endpoint of a \SRAP) which has $(i-2)$ more $(-1,-1)$ edges than $(+1,+1)$ edges. Hence the alternating path $\rho$ = $\rho_I \hspace{2mm} o \hspace{2mm} \rho_R$ has more $(+1,+1)$ edges than $(-1,-1)$ edges (where $o$ denotes the concatenation of two paths). Here $\rho$ is an alternating path which starts from $m_I$ then it goes to $m$ which has $i$ $(+1,+1)$ edges then it takes the edge $(m,w)$ which is a $(-1,-1)$ edge and then it takes the alternating path from $w$ to $m_R$ which has $(i-2)$ $(-1,-1)$ edges. Hence, $\rho$ has $i$ $(+1,+1)$ edges and $(i-1)$ $(-1,-1)$ edges. Hence, $M \oplus \rho$ is a more popular matching than $M$, which is a contradiction. Hence, $m \in A_m$ at level $i$ cannot be adjacent to a woman $w \in B_d$ at level $(i-1)$.

\textbf{Condition (iii)}: Suppose there is a man $m \in A_m$ at level $i$ which is adjacent to a woman $w \in B_d$ at level $i$ such that $(m,w)$ is not a $(-1,-1)$ edge. So, from Lemma \ref{1c} we get that there is an alternating path $\rho_I$ from $m$ to $m_I$ (where $m_I$ is an unmatched man and also an endpoint of a SIAP) with $i$ more $(+1,+1)$ edges than $(-1,-1)$ edges. Again from Lemma \ref{1c} we get that there is an alternating path $\rho_R$ from $w$ to $m_R$ (where $m_R$ is a non-critical man and also an endpoint of a SRAP) which has $(i-1)$ more $(-1,-1)$ edges than $(+1,+1)$ edges. Hence the alternating path $\rho$ = $\rho_I \hspace{2mm} o \hspace{2mm} \rho_R$ has more $(+1,+1)$ edges than $(-1,-1)$ edges (where $o$ denotes the concatenation of two paths). Here $\rho$ is an alternating path which starts from $m_I$ then it goes to $m$ which has $i$ $(+1,+1)$ edges then it takes the edge $(m,w)$ which is not a $(-1,-1)$ edge and then it takes the alternating path $w$ to $m_R$ which has $(i-1)$ $(-1,-1)$ edges. Hence, $\rho$ has $i$ $(+1,+1)$ edges and $(i-1)$ $(-1,-1)$ edges. Hence, $M \oplus \rho$ is a more popular matching than $M$, which is a contradiction. Hence, if $m \in A_m$ at level $i$ is adjacent to a woman $w \in B_d$ at level $i$ then $(m,w)$ is a $(-1,-1)$ edge.

\end{proof}
%----------------------------------------------------------------------------
%----------------------------Lemma 32-----------------------------------------
\begin{lemma} \label{4l}
For an edge $(m,w) \in A_r \times B_d$ in $G$ we have the following
\begin{enumerate}[(i)]
    \item If $m$ is at level $i$ and $w$ is at level $(i+1)$ then the edge $(m,w)$ is not a $(+1,+1)$ edge.
    \item If $m$ is at level $i$ then $w$ cannot be at level $(i-1)$ or below.
    \item If $m$ is at level $i$ and $w$ is at level $i$ then $(m,w)$ is a $(-1,-1)$ edge. 
\end{enumerate}
\end{lemma}
\begin{proof}
\textbf{Condition (i)}: If $m$ is at level $i$ and $w$ is at level $(i+1)$ such that the edge $(m,w)$ is a $(+1,+1)$ edge then in step (e) of the partition method $m$ and his matched partner are added to $A_d \cup B_d$. Hence $m \notin A_r$.
\\\\
\textbf{Conditions (ii) and (iii)}: These two conditions are vacuously true because according to construction of the sets defined in the partition method there are no edges $(m,w) \in A_r \times B_d$ such that $m$ is at level $i$ and $w$ is at level $j$ where $j \leq i$. This is because if level of $m$ is $i$ where $j \leq i$ then in step (e) of the partition method $m$ and his matched partner are added to $A_d \cup B_d$ 
\end{proof}

%------------------------Lemma 33----------------------------------------
\begin{lemma} \label{5l}
For a pair $(m,w) \in A_m \times B_r$ we have the following
\begin{enumerate}[(i)]
    \item If $m$ is at level $i$ and $w$ is at level $(i+1)$ then the edge $(m,w)$ is not a $(+1,+1)$ edge.
    \item If $m$ is at level $i$ then $w$ cannot be at level $(i-1)$ or below.
    \item If $m$ is at level $i$ and $w$ is at level $i$ then $(m,w)$ is a $(-1,-1)$ edge. 
\end{enumerate}
\end{lemma}

%-------------------------------------------------------------------------

%--------------------------------------------------------------------------------
%----------------------------------
\begin{proof}
\textbf{Condition (i)}: If $m$ is at level $i$ and $w$ is at level $(i+1)$ such that the edge $(m,w)$ is a $(+1,+1)$ edge then in step (f) of the partition method $w$ and her matched partner are added to $A_m \cup B_m$. Hence $w \notin B_r$.
\\\\
\textbf{Conditions (ii) and (iii)}: These two conditions are vacuously true because according to construction of the sets defined in the partition method there are no edges $(m,w) \in A_m \times B_r$ such that $m$ is at level $i$ and $w$ is at level $j$ where $j \leq i$. This is because if level of $w$ is $j$ where $j \leq i$ then in step (f) of the partition method $w$ and her matched partner are added to $A_m \cup B_m$ 
\end{proof}
\subsection{Transformation of $M_m$ to $M'_d$} \label{sec:algo-convert-to-Md}
%The steps involved are as follows, 
%Transformation $2$ converts $M_m$ to $M_d^*$.
\begin{enumerate}[(a)]
    \item All the men unmatched in $M_m$ are assigned level $1$ and they start proposing from the beginning of their preference lists. A woman prefers a man at level $j$ more than a man at level $i$ where $j > i$. If a man $m$ proposes a woman $w$ then $w$ will accept $m$'s proposal iff $w$ is unmatched or if $w$ prefers $m$ more than her matched partner.
    \item If a critical man $m$ at level $i$ where $i < |C|$ exhausts his preference list while proposing and remains unmatched then we assign level $(i+1)$ to $m$ and $m$ starts proposing again from the beginning of his preference list.
    \item If a non-critical man $m$ at level $0$ exhausts his preference list while proposing and remains unmatched then we assign level $1$ to $m$ and $m$ starts proposing again from the beginning of his preference list.
\end{enumerate}
Let $M'_d$ be the matching which we get after applying the above steps on the induced subgraph on $A_m \cup B_m$.

%-------------------------------------------------------------------
\subsection {Transformation Procedures}\label{sec:convert}
We prove Theorem~\ref{thm:convert} now.
Recall that, before partitioning $A\cup B$, we have assigned levels, denoted by $level(u)$ to all the vertices $u\in A\cup B$ according to $M$ using Algorithm~\ref{LAmPFM}, and that $level(u)=level(M(u))$. Our transformation procedures use these levels, 
%---------------Transformation 1-----------------------------------

\subsubsection{Transformation of $M_d$ to $M'_m$: }\label{sec:algo-convert-to-Mm}%old label 3.1s

Following are the steps involved in the transformation, we refer to this as {\em Transformation $1$}.

\begin{enumerate}[(a)]
    \item For $m\in A_d,M(m)\in B_d$, if $level(m)=level(M(m))=i$, $i \geq 1$, then set $level(m)=level(M(m))=i-1$%we assign level $(i-1)$ to $m$ and its matched partner $M(m)$.
    \item Mark the matched edges present among level $0$ vertices as unmatched edges. So all the level $0$ men in $A_d$ are not assigned to any partner now.
    \item Execute a proposal algorithm now. The men at level $0$ start proposing from the beginning of their preference lists. A woman prefers a man at level $j$ more than a man at level $i$ where $j > i$. If a man $m$ proposes to a woman $w$ then $w$ will accept $m$'s proposal iff $w$ is unmatched or if $w$ prefers $m$ more than her matched partner.  
    \item If a critical man $m$ at level $i$ where $i < |C|$ exhausts his preference list while proposing and remains unmatched then we assign level $(i+1)$ to $m$ and $m$ starts proposing again from the beginning of his preference list.
\end{enumerate}
    Let $M'_m$ be the matching obtained after applying the above steps on the induced subgraph on $A_d \cup B_d$, and let $M^*_m=M'_m\dot{\cup} M_m\dot{\cup} M_r$ be the resulting matching in $G$. %We refer to the above steps as {\em Transformation $1$}. 

\subsubsection{Transformation of $M_m$ to $M'_d$:} This is referred to as {Transformation $2$} here onwards, and involves promoting all the unmatched men to level $1$ and executing a similar proposal algorithm as above. Men that get unmatched during the course of the proposal algorithm continue proposing to women further down in their preference list. If they exhaust their preference list without getting matched, then they are promoted to the next higher level and continue proposing, however, non-critical men are not promoted beyond level $1$.% We give the formal transformation in Appendix. 
Let the resulting matching in $G$ be $M^*_d=M'_m\dot{\cup} M_d\dot{\cup} M_r$.

%----------------------Transformation 2-----------------------------

%-----------------------------------------------------------------------
The following property is crucially used in proving that $M^*_m$ is a \mPFM\ in $G$ whereas $M^*_d$ is a \DFM\ in $G$.
\begin{lemma}\label{6l}
For a man $m \in A_d$, if $level(m)=i, i>0$ before applying Transformation $1$ on $A_d \cup B_d$ then, after applying Transformation $1$, $level(m)\in\{i-1,i\}$. The same holds for a woman $w\in B_d$. Similarly, for a man $m \in A_m$, if $level(m)=i$ before applying Transformation $2$ on $A_m \cup B_m$ then, after applying Transformation $2$, $level(m)\in \{i,i+1\}$. The same holds for a woman $w\in B_m$.
\end{lemma}
\begin{proof}
We prove the property for Transformation $2$. The proof for Transformation $1$ is analogous.

Suppose there exists a man $m \in A_m$ who was assigned level $i$ before applying Transformation $2$ on $A_m \cup B_m$ but after applying Transformation $2$ suppose $m$ is assigned level $(i+2)$ or more. In Transformation $2$ we convert the matching $M_m$ to $M_d^*$. If $m$ is unmatched in $M_m$ then the level of $m$ in $M_m$ was 0 and it can be at most at level 1 in $M_d^*$ because $m$ is a non-critical man. So, $m$ cannot be an unmatched man because it contradicts our assumption that the level of $m$ in $M_d^*$ is $(i+2)$ or more. Suppose $m$ is matched to a woman $w$ in $M_m$ and the level of $m$ is $i$. While applying Transformation $2$ $w$ rejected $m$ because $w$ got a proposal from some man $m'$ who is better than $m$ and is at level $i$. Note that $w$ must have rejected $m$ while applying Transformation $2$ because after applying Transformation $2$ level of $m$ changes to $(i+2)$ or more from $i$. Now, since $m$ is assigned level $(i+2)$ or more so $m$ exhausts his preference list while proposing and  remains unmatched at level $i$. So, $m$ gets promoted to level $(i+1)$ and he starts proposing from the beginning of his preference list. Again since  $m$ is assigned level $(i+2)$ or more so $m$ exhausts his preference list while proposing and  remains unmatched at level $i+1$ but this is not possible because in the worst case $m$ can propose to $w$ and get matched to her. This is because $w$ would reject $m'$ which is at level $i$ and $m$ is at level $(i+1)$. Hence $w$ would accept $m$'s proposal and the level of $w$ changes to $(i+1)$. Note that no man $m''$ can get promoted from level $i$ to level $(i+1)$ and breaks the engagement of $m$  and $w$ because if this happens then in $M_m$  the edge $(m'',w)$ is a $(+1,+1)$ edge but since both $m''$ and $w$ are in the same level $i$ in $M_m$ hence $(m'',w)$ cannot be a $(+1,+1)$ edge. Hence, $m$ does not exhaust his preference list while proposing at level $i+1$. So, after applying Transformation $2$ the level of $m$ can be either $i$ or $(i+1)$.
\end{proof}

Theorems \ref{thm:mPFM-proof} and \ref{thm:DFM-proof} show that the matchings output by the transformations are a \mPFM\ and a \DFM\ in $G$ respectively. %The proofs involve proving that the resulting matchings satisfy the four conditions in Theorem~\ref{thm:mPFM}
%We prove Theorem~\ref{thm:DFM-proof} in Section~\ref{app:main-thm}.
%and Theorem~\ref{thm:DFM} respectively. %We defer the proof of Theor to Section \ref{app:main-thm}.
%-----------------------------Theorem 36-------------------------------------
\begin{theorem}\label{thm:mPFM-proof}
$M^*_m$ = $(M_m' \cup M_m \cup M_r)$ is a \mPFM.
\end{theorem}

\begin{proof}
The four conditions given in Theorem \ref{thm:mPFM} are sufficient to show that a matching is a \mPFM. We show that $M'$ satisfies all of them.

Before applying the Transformation $1$, $M$ satisfied conditions $1$ to $3$ of Theorem \ref{thm:mPFM} because of the way Algorithm~\ref{LAmPFM} assigns levels.]

After applying Transformation $1$, the matching $M_d$ changes to $M_m'$ and the levels of the vertices in $A_d \cup B_d$ decrease by at most $1$ (Lemma \ref{6l}). So, if $M'$ does not satisfy conditions $1-3$ of Theorem~\ref{thm:mPFM} then it has to be because of the pairs present in $A_m \times B_d$ and $A_r \times B_d$. Now we show that the conditions are still satisfied. 

Below the proofs are given only for the pairs in $A_m \times B_d$. Proofs for the pairs in $A_r \times B_d$ are similar.% to the proofs given for the pairs in $(A_m \times B_d)$.

Let $(m,w)$ be an edge in $A_m\times B_d$.
From Lemma \ref{3l} (ii), if the level of $w$ is $i$ with respect to $M$, then $m$ has level $j\leq i$. Now, after applying the Transformation 1, the level of $w$ either remains $i$ or becomes $(i-1)$. In the former case, 
the first condition of Theorem \ref{thm:mPFM} is satisfied.
In the later case, we have three possibilities: (a) either $j<(i-1)$ or (b) $j=(i-1)$, or (c) $j=i$. In case (b), $(m,w)$ is  not a $(+1,+1)$ edge (Lemma \ref{3l} (i)), in case (c), $(m,w)$ is a $(-1,-1)$ edge (Lemma \ref{3l} (iii)). Hence, there is no $(+1,+1)$ edge in between a pair $(m,w) \in A_m \times B_d$ in the matching $M'$ where $m$ is at level $j$ and $w$ is at level $i$ such that $j \leq i$. Hence, condition $1$ is satisfied.

From Lemma \ref{3l} (ii), if the level of $w$ is $i$ in $M$, then $m$ has level $j\leq i$. 
If level of $w$ changes to $(i-1)$ after applying the Transformation $1$, and if level of $m$ is $i$, then due to Lemma \ref{3l}(iii), $(m,w)$ is a $(-1,-1)$ edge. Thus the condition $2$ of Theorem \ref{thm:mPFM} is satisfied.

From Lemma \ref{3l} (i), if $w$ is at level $i$ with respect to $M$, then level of $m$ is $j\leq i$. 
If level of $w$ changes to $(i-1)$, the conditions of Theorem \ref{thm:mPFM} are still satisfied because no man in $A_m$ adjacent to $w$ is present at level $(i+1)$ or above.

We know that all the unmatched men are non-critical men. In the first step of Transformation $1$, we decrease the level of each vertex by $1$. Since the level of a non-critical man is at most $1$ to begin with, and they are never promoted to a higher level in the Transformation $1$, 
all the vertices unmatched in $M^*_m$ remain at level $0$.  
%\end{enumerate}
Since all the conditions of Theorem \ref{thm:mPFM} are satisfied, $M^*_m$ is a \mPFM.
\end{proof}

The following is an analogous result for Transformation $2$.
%--------------------------------Lemma 34------------------------------

%------------------------------------------------------------------------------------------------------------Theorem 37----------------------------------
\begin{theorem}\label{thm:DFM-proof}
$M^*_d$ = $(M_d \cup M'_d \cup M_r)$ is a \DFM\ in $G$.
\end{theorem}
\begin{proof}%[Proof of Theorem \ref{thm:DFM-proof}]
Recall the 4 conditions given in Theorem $\ref{thm:DFM}$ which were sufficient to show that a matching is a \DFM\ . So, now, we will show that $M^*_d$ satisfies all the 4 conditions given in Theorem $\ref{thm:DFM}$. Hence, $M^*_d$ is a \DFM.

Before applying Transformation $2$ the conditions 1 to 3 of Theorem \ref{thm:DFM} were already satisfied in the matching $M$ because the levelling algorithm for \mPFM\  assigns levels to the vertices in such a manner that the conditions 1 to 3 gets satisfied (recall the three phases of an iteration of the Levelling Algorithm for \mPFM\ , each phase ensures each condition from 1 to 3 of Theorem \ref{thm:DFM} gets satisfied). But after applying Transformation $2$ the matching $M_m$ changes to $M_d^*$ and the level of the vertices present in $A_m \cup B_m$ increases by at most 1 (Lemma \ref{6l}). So, if the conditions 1 to 3 of Theorem \ref{thm:DFM} are not satisfied in $M^*_d$ then it has to be because of the pairs present in $A_m \times B_d$ and $A_m \times B_r$. So, now we show that the pairs in $A_m \times B_d$ and $A_m \times B_r$ in the matching $M^*_d$ will also satisfy the conditions 1 to 3. Below the proofs are given only for the pairs present in $A_m \times B_d$. Proofs for the pairs present in $A_m \times B_r$ are similar to the proofs given for the pairs in $(A_m \times B_d)$.

\noindent\textbf{Condition $1$}: For a pair $(m,w) \in A_m \times B_d$ if $(m,w)$ is an edge in $G$ then from Lemma \ref{3l} (ii) we get that before applying Transformation $2$ if the level of $m$ is $i$ then $w$ is present at an level $i$ or higher. Now, after applying Transformation $2$ the level of $m$ either remains $i$ or becomes $(i+1)$. So, if the level of $m$ remains $i$ after applying Transformation $2$ then the Condition $1$ of Theorem \ref{thm:DFM} is satisfied because $w$ is present either at a level higher than $i$ or at level $i$ in that case the edge $(m,w)$ is a $(-1,-1)$ edge (Lemma \ref{3l} (iii)). Now, if the level of $m$ becomes $(i+1)$ after applying Transformation $2$ then we have three cases either $w$ is at level higher than $(i+1)$ or $w$ is at level $i$ in that case $(m,w)$ is a $(-1,-1)$ edge (Lemma \ref{3l} (iii)) or $w$ is at level $(i+1)$ in that case $(m,w)$ is not a $(+1,+1)$ edge (Lemma \ref{3l} (i)). Hence, there is no $(+1,+1)$ edge in between a pair $(m,w) \in A_m \times B_d$ in the matching $M^*_d$ where $m$ is at level $i$ and $w$ is at level $j$ such that $j \leq i$. Hence, Condition $1$ is satisfied.
\\\\
\textbf{Condition $2$}: For a pair $(m,w) \in A_m \times B_d$ if $(m,w)$ is an edge in $G$ then from Lemma \ref{3l} (ii) we get that before applying Transformation $2$ if the level of $m$ is $i$ then $w$ is present at an level $i$ or higher. Hence, no man at level $i$ is adjacent to a woman at level $(i-1)$. Hence, Condition $2$ is vacuously satisfied when $m$ remains at level $i$ after applying Transformation $2$. If level of $m$ changes to $(i+1)$ after applying Transformation $2$ and if level of $w$ is $i$ then due to Lemma \ref{3l} (iii) we get that $(m,w)$ is a $(-1,-1)$ edge. Hence, Condition $2$ is satisfied.
\\\\
\textbf{Condition $3$}:  For a pair $(m,w) \in A_m \times B_d$ if $(m,w)$ is an edge in $G$ then from lemma \ref{3l} (i) we get that before applying Transformation $2$ if $m$ is at level $i$ then level of $w$ is $i$ or more. If level of $m$ remains $i$ after applying Transformation $2$ then Condition $3$ of Theorem \ref{thm:DFM} is satisfied because no woman in $B_d$ adjacent to $m$ is present at level $(i-2)$ or below. If level of $m$ changes to $(i+1)$ then also Condition $3$ of Theorem \ref{thm:DFM} is satisfied because no woman in $B_d$ adjacent to $m$ is present at level $(i-1)$ or below.
\\\\
\textbf{Condition $4$}: We know that the set of unmatched men are non-critical men. While applying Transformation $2$ if a non-critical man $m$ exhausts its preference list while proposing and remains unmatched at level 0 then we assign level 1 to $m$ and $m$ starts proposing again from the beginning of his preference list. Now if $m$ again exhausts his preference list while proposing and remains unmatched at level 1 then $m$ remains unmatched in the matching $M^*_d$. Hence, if a man $m$ is unmatched in $M^*_d$ it has to be at level 1. Hence Condition $4$ of Theorem \ref{thm:DFM} is satisfied.
\\\\
Since all the conditions of Theorem \ref{thm:DFM} are satisfied. Hence, $M^*_d$ is a \DFM\ .
\end{proof}

%% file: 4-Appendix.tex
\newpage
\section{Proofs from Section \ref{sec:algo-mPFM-DFM}}\label{A}
%-----------------------Theorem 7-------------------------------------

%--------------------------DFM part-------------------------
\subsection{Reduction for \DFM}
The high level idea to find a \DFM\ in an \SMPM\ is exactly the same as finding the \mPFM\ . The results given in this section are similar to the results given for \mPFM\ . 
At first we reduce our \SMPM\ $G = (A \cup B)$ to a stable marriage instance (\SMtwo) $G'' = (A'' \cup B'',E'')$. Then we show that every stable matching in $G''$ can be mapped to a \DFM\ in $G$. %Hence we find a stable matching in the SM2 instance using Gale-Shapley Algorithm and map it to a \DFM\ in the \SMPM\. So, in the following section we give the reduction from a \SMPM\ to a stable marriage instance SM2. 
We also show surjectivity of this map. %While this reduction appears in \cite{NN17}

The reduction for \DFM\ is very similar to that for \mPFM. The only difference with the previous reduction is the following. In the reduction for \mPFM,
$G'$ has only one copy of a man $m\in A\setminus P$. In this reduction, such men have two copies in the reduced instance $G''$, and there is one dummy
woman in $G''$ corresponding to such a man. The number of levels in this reduction is one more than that for \mPFM.
%-----------------------------------------------------------------------------

\subsubsection{Reduction}
\begin{itemize}
\item \textbf{The set $A''$: }
Let $\ell$ be the number of men in $G$ who have privileges i.e. $\ell=|P|$. For a man $m \in P$, $A''$ has $(\ell+2)$ copies of $m$, denoted by $m^0,m^1,...,m^{\ell+1}$. Let $A''_m$ denote the set of copies of $m$ in $A''$. We refer to $m^i \in A''$ as the level $i$ copy of $m \in A$. For a man $m\in A \setminus P$, $A''$ has only two copies of $m$, thus $A''_m=\{m^0, m^1\}$. Now, $A''= \bigcup\limits_{m \in A} A''_m $.

\item \textbf{The set $B''$: }
All the women in $B$ are present in $B''$. Additionally,
corresponding to a man $m \in P$, $B''$ contains $\ell + 1$ dummy women $d_m^1,d_m^2,d_m^3,...,d_m^{\ell+1}$, denote the set of these women as $D_m$. We call the dummy woman $d_m^i$ as the level $i$ dummy woman for $m$. There is one dummy woman $d_m^1$ corresponding to a man $m$ in $A\setminus P$. Now, $B'' = B \cup \bigcup\limits_{m \in A} D_m$. 
\end{itemize}
We denote by $\langle list_m \rangle$ and $\langle list_w \rangle$ the preference lists of $m \in A$ and $w \in B$ respectively. Let $\langle list_w \rangle ^ i$ be the list of level $i$ copies of men present in $\langle list_w \rangle$. Note that, for $m\in A\setminus P$ present in \listw, 
the level $i$ copy of $m$ for $i \geq 2$ is not present in $A''$. Then $\langle list_w \rangle ^ i$ does not contain the level $i$ copy of that man for $i \geq 2$. We now describe the preference lists in $G''$. Here $\circ$ denotes the concatenation of two lists. % of a man $m^i \in A'$ , a non dummy woman $w \in B'$ and a dummy woman $d^i_m \in B'$.

\begin{center}
		\begin{tabular}{lll}
		$m\in A\setminus P$: & &\\
		$m^0$ &: & $\langle list_m \rangle\circ d_m^1$\\	
		$m^1$ &: & $d_m^1\circ \langle list_m \rangle$\\	
		$m\in P$, $i \in \{0,\ell+1\}$: & & \\
		$m^0$ &:& $\langle list_m \rangle$, $d_m^1$\\
		$m^i$ &: & $d_m^i$, $\langle list_m \rangle$, $d_m^{i+1}$, $i \in \{1,\ell\}$)\\
		$m^{\ell+1}$ &:& $d_m^{\ell+1}$, $\langle list_m \rangle$\\
		$w$ s.t. $w\in B$ &: & $\langle list_w \rangle ^ {\ell+1} \circ \langle list_w \rangle ^ {\ell} \circ\ldots\circ \langle list_w \rangle^0$ \\
		$d^i_m$,$i \in \{1, \ell+1\}$ &:& $m^{i-1}$, $m^i$ \\
		\end{tabular}
	\end{center}
We refer to the instance $G''$ as \SMtwo.

%-----------------------------------------------------------------------------

\subsubsection{Correctness of the reduction}
After constructing the instance \SMtwo, our goal is to map a stable matching $M''$ in \SMtwo\ to a \DFM\ $M$ in $G$. The mapping
is a simple and natural one: For a man $m \in A$, define $M(m) = B \cap \bigcup\limits_{m \in A} M''(m^i)$. Note that $M(m)$ denotes the set of non-dummy women who are matched to any copy of $m$ in $A''$. In the rest of this section, the term {\em image} always refers to the image under this map. A man $m \in A$ in the matching $M$ is unmatched if none of its copies in $A''$ gets matched to a non-dummy woman in the matching $M''$. It remains to prove that $M$ is a \DFM. This involves showing that $M$ is a matching, it is feasible, popular, and all the matchings larger than $M$ does not get strictly more votes than $M$. 

To show that $M$ is a matching in $G$, we need to prove %that in, any stable matching $M'$ in the \SMone\ instance, at most one copy of a man $m \in P$ gets matched to a non-dummy woman. 
the following theorem.
The proof uses the fact that there are $\ell+2$ copies of a man $m \in P$, and $\ell + 1$
dummy women corresponding to that man $m \in P$, and each dummy woman is the first choice of some copy of $m$.

%----------------------------Theorem 11---------------------------------------

\begin{theorem}\label{thm:DFM-matching}
In any stable matching $M''$ in the \SMtwo\ instance, at most one copy of a man $m \in A$ gets matched to a non-dummy woman.    
\end{theorem}
\begin{proof}
Suppose $m^i$ be the copy of the man $m \in A$ which gets matched to a non-dummy woman. We prove using induction on $j$ where $j > i$ that $m^j$ does not get matched to a non-dummy woman and it gets matched to the dummy woman $d_m^j$.\\ 
\underline{Base Case}: Suppose $j = (i+1)$. According to the theorem $m^i$ gets matched to a non dummy woman. Now, if $m^j$ is not matched to $d_m^j$ then $d_m^j$ remains unmatched in $M''$ because $m^i$ and $m^j$ (note that $j = (i+1)$) are the only vertices adjacent to $d_m^j$. In that case the edge $(m^j,d_m^j)$ forms a $(+1,+1)$ edge in $M''$ because $m^j$ prefers $d_m^j$ the most. This is a contradiction as $M''$ is a stable matching. Hence, $m^j$ is matched to $d_m^j$.\\
\underline{Inductive Hypothesis}: Assume that for  $j = k$ (where $k > i$) $m^j$ gets matched to $d_m^j$\\
\underline{Inductive Step}: Now for $j = (k+1)$ we prove that $m^j$ gets matched to the dummy woman $d_m^j$. According to the inductive hypothesis $m^{j-1}$ gets matched to the dummy woman $d_m^{j-1}$ (note that $j = (k+1)$ here). Now, if $m^j$ is not matched to $d_m^j$ then $d_m^j$ remains unmatched in $M''$ because $m^{j-1}$ and $m^j$ are the only vertices adjacent to $d_m^j$. In that case the edge $(m^j,d_m^j)$ forms a $(+1,+1)$ edge in $M''$ because $m^j$ prefers $d_m^j$ the most. This is a contradiction as $M''$ is a stable matching. Hence, $m^j$ is matched to $d_m^j$.

Hence, by induction we get that all the copies of $m$ whose levels are greater than $i$ cannot get matched to a non-dummy woman and gets matched to the first dummy woman in his preference list. Since $i$ can be anything in the range $[0,\ell + 1]$,hence we proved that at most one copy of a man $m \in A$ gets matched to a non dummy woman.
\end{proof}

Theorem \ref{thm:DFM-matching} shows that $M$ is a matching. Now, we prove the following two corollaries.

%------------------------------Corollary 12-----------------------------------

\begin{corollary}\label{3.1.1}
If in a stable matching in the \SMtwo\ instance $m^i$ (where $i \in [0,\ell+1]$) is matched to a non dummy woman then $m^j$ (where $j > i)$ is matched to the dummy woman $d_m^j$ which is the first dummy woman in the preference list of $m^j$. 
\end{corollary}
\begin{proof}
The proof of this would be similar to the inductive proof done in the Theorem \ref{thm:DFM-matching} \end{proof}

%-----------------------------------------------------------------------------

%-------------------------------Corollary 13----------------------------------
\begin{corollary}\label{3.1.2}
If in a stable matching in the \SMtwo\ instance $m^i$ (where $i \in [0,\ell+1]$) is matched to a non dummy woman then $m^j$ (where $j < i)$ is matched to the dummy woman $d_m^{j+1}$ which is the last dummy woman in the preference list of $m^j$. 
\end{corollary}
\begin{proof}
If $i = 0$ then the statement is vacuously true. So, if $i \neq 0$ then we get $m^0$ has to be matched to a dummy woman(from Theorem \ref{thm:DFM-matching}). Now since there is no dummy woman present in the beginning of the preference list of $m^0$, so $m^0$ has to be matched with $d_m^1$ which is the last dummy woman in the preference list of $m^0$. Again since $m^0$ is matched to $d_m^1$ so $m^1$ cannot be matched to $d_m^1$, it has to be matched with $d_m^2$ which is again the last dummy woman in the preference list of $m^1$. This will continue up to $m^i$ which gets matched to a non dummy woman.    
\end{proof}
%-----------------------------------------------------------------------------

%-------------------------Theorem 14------------------------------------------

\begin{theorem}\label{thm:DFM}
A matching $M$ in $G$ that is an image of a stable matching $M''$ in $G''$ is a \DFM\ if it satisfies the following conditions. Moreover,
every such matching satisfies the conditions.
\begin{enumerate}
    \item All $(+1,+1)$ edges are present in between a man at level $i$ and a woman $w$ at level j where $j > i$.
    \item All edges between a man at level $i$ and a woman at level $(i-1)$ are $(-1,-1)$ edges.
    \item No edge is present in between a man at level $i$ and a woman at level $j$ where $j \leq (i-2)$.
    \item All unmatched men are at level 1.
\end{enumerate}
\end{theorem}
%-----------------------------------------------------------------------------
Here we need to show that if $M$ satisfies the four conditions, then $M$ is a \DFM. So, to prove this we show at first $M$ is a \PFM\ and then we show that, for all feasible matchings $N$ such that $|N| > |M|$, we have $\phi(N,M) < \phi(M,N)$.
\\\\
\textbf{$\mathbf{M}$ is a \PFM\ }: To show $M$ is a \PFM\ at first we need to show that $M$ is a feasible matching and then we need to prove that for any other feasible matching $N$ in $G$ we have $\phi(N,M) \leq \phi(M,N)$.
%——————Proving M is a feasible Matching———————

So, now we prove $M$ is a feasible matching. The arguments are similar to the arguments given in the case of \mPFM\ .

Suppose $M$ is not a feasible matching and there exist a feasible matching $N$ in that marriage instance with critical men. Recall that we are only concerned with those marriage instance with critical men which has at least one feasible matching. Suppose $m$ be a critical man who is unmatched in $M$. So, the graph $M \oplus N$ must contain an alternating path $\rho$ which starts from $m$. Now $\rho$ can end in a man $m'$ or in a woman $w'$. \underline{CASE 1}: $\rho$ ends in $m'$: Let $\rho = (m,w,m_1,w_1,....,m)$. Since $\rho$ ends in $m'$, hence $m'$ must be unmatched in $N$. Since $N$ is a feasible matching $m'$ must be a non critical man and hence will be at level 0 or at level 1. Since $m$ is unmatched in $M$ it has to be in the level $\ell+1$ otherwise if $m$ is at level $i$ where $i < (\ell+1)$ then $(m^i,d_m^{i+1})$ would be a $(+1,+1)$ edge in $M''$ because $m^i$ is unmatched in $M''$ and $d_m^{i+1}$ prefers $m^i$ the most in $G''$. Again no woman $w$ which is adjacent to $m$ can be at level $\ell$ or below because then $(m^{\ell+1},w)$ would form a $(+1,+1)$ edge in $M''$ as $m^{\ell+1}$ is unmatched and $w$ prefers $m^{\ell+1}$ more than her matched partner which is at level $\ell$ or below. Hence, in $\rho$, $w$ is at level $\ell+1$ again $M(w) = m_1$ is also at level $\ell+1$ because the level of a woman and her matched partner are same. Now, $w_1$ cannot be at level strictly less than $\ell$ due to Condition $3$ of Theorem \ref{thm:DFM}. Hence the alternating path $\rho$ can go only one level down that is from a man at level $i$ to a woman at level $i-1$. Note that all the men who are at level greater than 1 are critical men because there is no copy of a non-critical man of level greater than 1 in $G''$. Since $\rho$ can go only one level down, hence there must exist at least one critical man at each level from $2$ to $\ell$ and there are at least two critical men ($m$ and $m_1$) at level $\ell+1$. Hence, the number of critical men in $G$ is at least $\ell+1$. This is a contradiction because we know the number of critical men in $G$ is $\ell$. \underline{CASE $2$}: $\rho$ ends in $w'$. Since $\rho$ ends in $w'$ it has to unmatched in $M$ and thus the level of $w$ is 0 as the level of each unmatched woman is given 0. Hence $\rho$ starts from a man at level $\ell$ and ends at level 0. Since $\rho$ can only go one level down, hence using the same arguments as used in case 1 we get that there are at least $\ell+1$ critical men in $G$. This is a contradiction because we know the number of critical men in $G$ is $\ell$. Hence $M$ is a feasible matching.

Now, we prove $M$ is a \PFM.

\textbf{$\mathbf{M}$ is a \PFM}: Proof of this part is exactly same as the proof given for theorem \ref{thm:mPFM}.
\\\\
\textbf{$\mathbf{M}$ is a \DFM}: Now we show that for any feasible matching $N$ such that $|N| > |M|$ we have $\phi(N,M) < \phi(M,N)$. We take the graph $M \oplus N$, which is the disjoint union of alternating paths and cycles. There is no alternating path or cycle $\rho$ in $M \oplus N$ such that $\phi((M \oplus \rho),M) > \phi(M,(M \oplus \rho))$ otherwise $M$ is not a \PFM\. So, now we need to show an alternating path or cycle in $M \oplus N$ such that $\phi((M \oplus \rho),M) < \phi(M,(M \oplus \rho))$ then only we can say $\phi(N,M) < \phi(M,N)$. Now, since $|N| > |M|$  there must exist an alternating path which starts from a man $m_1$ unmatched in $M$ and ends in a woman $w$ unmatched in $M$. Since $m_1$ is unmatched in $M$ it is at level 1 due to Condition $4$ of theorem \ref{thm:DFM}. Suppose $\rho = (m_1,w_1,m_2,w_2,m_3,w_3...,w)$. The level of $w$ is 0 as all unmatched women in $M$ are defined to be at level 0. Let $j$ be the highest level of a man present in $\rho$. Note that the edges $(m_i,w_i)$ are all edges present in $N$. Since $m_1$ is at level 1 and the highest level of a man in $\rho$ is $j$, hence we would have at most $(j-1)$ $(+1,+1)$ edges from $m_1$ to the $j^{th}$ level man as due to Condition $1$ of theorem \ref{thm:DFM} we get that $(+1,+1)$ edges are only present in between a lower level man and a higher level woman. Now, from Condition $3$ of theorem \ref{thm:DFM} we get that $\rho$ can go only one level down. Due to condition $2$ we get that while going $\rho$ can only take $(-1,-1)$ edges. Since level of $w$ is 0, hence $\rho$ must have $j$ $(-1,-1)$ edges. Since $j \geq (j-1)$, hence the number of $(-1,-1)$ edges is strictly greater than the number of $(+1,+1)$ edges in $\rho$. Hence,  $\phi((M \oplus \rho),M) < \phi(M,(M \oplus \rho))$. Hence, $M$ is a \DFM.

Now we show that any $M$ that is an image of a stable matching $M''$ in $G''$ satisfies all the four conditions.
%----------------------Proof of four conditions-------------------------------
\underline{Condition $1$}: Suppose there is $(+1,+1)$ edge in between a man $m$ at level $i$ and woman $w$ at level $j$ such that $j \leq i$ in the matching $M$. Hence $m$ prefers $w$ more than his matched partner in $M$. Now, $M''(m^i) = M(m)$ and since the preference list of $m^i$ in  $G''$ is same as the preference list of $m$ in the \SMPM\ (except the dummy women in the beginning and end of the preference list of $m^i$), $m^i$ prefers $w$ more than $M''(m^i)$. So, in $M''$ the edge $(m^i,w)$ will be a $(+1,+1)$ edge because $m^i$ prefers $w$ more than $M''(m^i)$ and $w$ prefers $m^i$ more than $M''(w)$ because her matched partner is at level $j$ and $j \leq i$. In the SM2 instance $w$ prefers a level $i$ man more than a level $j$ man if $i > j$ and if $i = j$ then $w$ prefers $m^i$ more than $M''(w)$ because $w$ prefers $m$ more than $M(w)$ in the matching $M$. This contradicts the fact that $M''$ is stable matching. Hence, $M$ satisfies Condition $1$.\\
%-----------------------------------------------------------------------------
\underline{Condition $2$}: Suppose there is a man $m$ at level $i$ which is adjacent to a woman at level $(i-1)$ but the edge $(m,w)$ is not labelled $(-1,-1)$. $(m,w)$ cannot be labelled $(+1,+1)$ due to Condition $1$. So, it has to be labelled $(+1,-1)$ and $(-1,+1)$. CASE 1: If $(m,w)$ is labelled $(+1,-1)$ then $m$ prefers $w$ more than $M(m)$. Hence $m^i$ prefers $w$ more than $M''(m^i)$ and $w$ prefers $m^i$ more than its matched partner in $M''$ which is the $(i-1)$ level copy of $M(w)$. Hence the edge $(m^i,w)$ is a $(+1,+1)$ edge in the matching $M''$. This contradicts the stability of $M''$. Case $2$: Now, if $(m,w)$ is labelled $(-1,+1)$ then $w$ prefers $m$ more than $M(w)$. Now since $m$ is at level $i$ so $m^i$ gets matched to a non dummy woman in the matching $M''$. So, from Corollary \ref{3.1.2} we get that $m^{i-1}$ is matched to the dummy woman $d_m^i$ which is present at the end of his preference list. In this the edge $(m^{i-1},w)$ would be labelled $(+1,+1)$ because $m^{i-1}$ would prefer $w$ more than its matched partner in $M''$ which is present at the last of his preference list and $w$ would prefer $m^{i-1}$ more than $M''(w)$, which is a $(i-1)$ level copy of $M(w)$ as $w$ prefers $m$ more than $M(w)$. This again contradicts that $M''$ is a stable matching. Hence $M$ satisfies Condition $2$.\\
%-----------------------------------------------------------------------------
\underline{Condition $3$}: Suppose Condition $3$ is not satisfied, then there is a man $m$, which at level $i$ is adjacent to a woman $w$ at level $j$ such that $j \leq (i-2)$. In this case the edge $(m^{i-1},w)$ would be a $(+1,+1)$ edge because $m^{i-1}$ prefers $w$ over its matched partner in $M''$ which is $d_m^i$ (Corollary \ref{3.1.2}) and $w$ prefers $m^{i-1}$ over $M''(w)$ which is a $(i-2)$ level copy of $M(w)$. This contradicts the fact that $M''$ is a stable matching. Hence, $M$ satisfies Condition $3$.\\
%-----------------------------------------------------------------------------
\underline{Condition $4$}: Suppose $M = g(M'')$ is  a feasible matching. So, if there are unmatched men in $M$ then they are the non critical men. Let, $m$ be an arbitrary unmatched man in $M$. Now, during the reduction from $M$ to $M''$ we made two copies of $m$  in $A''$, they are $m^0$ and $m^1$. Since $m$ is unmatched in $M$ we have that one of $m^0$ or $m^1$ is unmatched in $M''$ and the other would get matched to the dummy woman $d_m^1$. Now, if $m^0$ is unmatched in $M''$ then the edge $m^0,d_m^1$ would form a $(+1,+1)$ edge because in $m^0$ is unmatched and $d_m^1$ prefers $m^0$ the most. This contradicts the fact that $M''$ is a stable matching. Hence, $m^1$ is unmatched in $M''$ and thus the level of all unmatched men in $M$ is 1.
\\\\
%--------------------End of proving 4 conditions------------------------------
Hence, any matching $M$ which is an image of a stable matching $M''$ in $G''$ is a \DFM.	
\\
The proof of Theorem \ref{thm:DFM} is similar to the proof of Theorem \ref{thm:mPFM}. Here we need to show that if a matching $M$ satisfies the four conditions given in theorem \ref{thm:DFM} then $M$ is a \DFM. So, to prove this we show at first $M$ is a \PFM\ and then we show that for all feasible matchings $N$ such that $|N| > |M|$ we have $\phi(N,M) < \phi(M,N)$.

\subsubsection{Surjectivity of the map}
%-------------------------Theorem 22-------------------------------------------
\begin{theorem}\label{thm:DFM-surjectivity}
For every \DFM\ $M$ in $G$, there exists a stable matching $M''$ in $G''$ such that $M$ is the image of $M''$.
\end{theorem}
\begin{proof}
At first we apply a {\em leveling algorithm} to $M$ (Algorithm~\ref{LADFM}). Algorithm~\ref{LADFM} takes a \mPFM\ $M$ as input and assigns levels to the vertices in $G$. The levels are used to get the pre-image of $M$ i.e. a stable matching $M''$ in $G''$. Once the levels are assigned by Algorithm~\ref{LADFM}, the pre-image is obvious - if a man $m$ in $G$ gets assigned to level $i$, and $M(m) = w$ then $M''(m^i) = w$. For $j < i$, $M''(m^j)=d_m^{j+1}$ which is the least preferred dummy woman on his list, and for $j > i$, $M''(m^j)=d_m^j$ which is the most preferred dummy woman on his list. 

The proof of the theorem is immediate from the correctness of Algorithm~\ref{LADFM}, proved below.
\end{proof} 
Now, we describe the leveling Algorithm which takes a \DFM\  as input.\\\\
%---------------------------------------------------------------------------------
\newpage

\begin{algorithm}
    \hspace*{\algorithmicindent}\textbf{Input:} A \DFM\  $M$ in a \SMPM\  instance.
    \\
    \hspace*{\algorithmicindent}\textbf{Output:} Assigns level to the vertices in $\mathcal{G}$ based on the matching $M$. 
    \begin{algorithmic}[1]
    \State {Initially all the unmatched men are assigned level 1 and all the vertices other than unmatched men are assigned level 0}
    \State flag = true
    \While{flag = true}
        \State check1 = 0, check2 = 0, check3 = 0
        \While{there exists a man $m$ at level $i$ and a woman $w$ at level $j$ such that $j \leq i$ and $(m,w)$ is a $(+1,+1)$ edge}
        \State Change the level of $w$ and its matched partner $M(w)$ from level $j$ to level $(i+1)$. Note that $w$ cannot be unmatched in $M$ because then $M$ would not be a PFM.
        \State check1 = 1
        \EndWhile
        \While{there exists a man $m$ at level $i$ and a woman $w$ at level $j$ such that $j < i$ and $(m,w)$ is a $(+1,-1)$ or a $(-1,+1)$ edge}
        \State Change the level of $w$ and its matched partner $M(w)$ from level $j$ to level $i$. Note that $w$ cannot be unmatched in $M$ because then $M$ would not be a PFM.
        \State check2 = 1
        \EndWhile
        \While {there exist a man $m$ at level $i$ and a woman $w$ at level $j$ such that $j \leq (i-2)$ and $(m,w)$ is a $(-1,-1)$ edge}
        \State Change the level of $w$ and its matched partner from level $j$ to level $(i-1)$. Note that $w$ cannot be unmatched in $M$ because then $M$ would not be a PFM.
        \State check3 = 1
        \EndWhile
        \If {check1 = 0 and check2 = 0 and check3 = 0}
        \State flag = false
        \EndIf
    \EndWhile
    \end{algorithmic}\caption{leveling Algorithm for \DFM\ }\label{LADFM}
\end{algorithm}

%-----------------------------------------------------------------------------
The proof for the termination of Algorithm~\ref{LADFM}  and the proof of $M''$ is a stable matching is exactly the same as the proof of the termination of Algorithm~\ref{LAmPFM}  and the proof of $M'$ is a stable matching respectively.
\begin{theorem}
All non-critical men are assigned level zero or one and the critical men are assigned level less than or equal to $(l+1)$ (where $l$ is the number of critical men in the \SMPM\  instance).
\end{theorem}
\begin{proof}
Suppose there is a non-critical man $m^i$ gets a level $i > 1$ then there is an alternating path from a woman $w^0$ at level 0 or from an unmatched man $m^0$ at level 0 (theorem \ref{thm:mPFM-term})to $m^i$ which has i more $(+1,+1)$ edges than $(-1,-1)$ edges. Let this alternating path be $\rho$. In this case $M \oplus \rho$ is a more popular matching than $M$. Hence contradiction. Hence, all non critical men are assigned level less than or equal to 1.

Suppose there is a critical man who is assigned level $j$ where $j > (l+1)$. Now, since there are only $l$ critical men in the \SMPM\  instance. Hence, there exists a level $i < j$ such that $i$ is empty. This contradicts theorem \ref{thm:mPFM-term}. Hence, all critical men are assigned level less than or equal to $(l+1)$.
\end{proof}